\documentclass[10pt,journal,compsoc]{IEEEtran}

\ifCLASSOPTIONcompsoc
  \usepackage[nocompress]{cite}
\else
  \usepackage{cite}
\fi

\usepackage[utf8]{inputenc}
\usepackage[english]{babel}
\usepackage{amsmath}
\usepackage{amsfonts}
\usepackage{amssymb}
\usepackage{color} 
\usepackage{bm}
\usepackage{listings}
\usepackage{caption}
\usepackage{amssymb}
\usepackage{amsthm}
\usepackage{graphicx}
\usepackage{epstopdf}
\usepackage{listings}
\usepackage{hyperref}
\usepackage{float}
\usepackage{amsmath}
\usepackage{amssymb}
\usepackage{amsfonts}
\usepackage{epstopdf}
\usepackage[overload]{empheq}

\usepackage{multirow}
\usepackage{amscd}
\usepackage{mathrsfs}
\usepackage{graphicx}
\usepackage{color}
\usepackage{url}
\usepackage{bm}
\usepackage{setspace}
\usepackage[font=small]{caption}
\usepackage{algorithm2e}
\usepackage{footnote}
\DeclareMathOperator*{\argmax}{arg\,max}
\newcommand*{\Scale}[2][4]{\scalebox{#1}{$#2$}}%
\newtheorem{theorem}{Theorem}
\newtheorem{lemma}{Lemma}
\newtheorem{fact}{Fact}
\newtheorem{definition}{Definition}
\newtheorem{proposition}{Proposition}
\newtheorem{corollary}{Corollary}

\ifCLASSINFOpdf
 
\else
  
\fi

\begin{document}
%
\title{A Two-Stage Auction Mechanism for Cloud Resource Allocation}

\author{Seyyedali Hosseinalipour,~\IEEEmembership{Student Member, IEEE,}
        Huaiyu Dai,~\IEEEmembership{Fellow, IEEE}
\IEEEcompsocitemizethanks{\IEEEcompsocthanksitem S. Hosseinalipour and H. Dai  are with the Department
of Electrical and Computer Engineering, North Carolina State University, Raleigh,
NC, USA.\protect\\
E-mail: shossei3@ncsu.edu, hdai@ncsu.edu
\IEEEcompsocthanksitem A preliminary version of this paper appeared in the proceeding of IEEE International Conference on Communications (ICC) 2017~\cite{ref:Mywork}.}}

\IEEEtitleabstractindextext{%
\begin{abstract}
The contemporary literature on cloud resource allocation is mostly focused on studying the interactions between customers and cloud managers. Nevertheless, the recent growth in the customers' demands and the emergence of private cloud providers (CPs) entice the cloud managers to rent extra resources from the CPs so as to handle their backlogged tasks and attract more customers. This also makes studying the interactions between the cloud managers and the CPs essential. In this paper, we investigate both of the mentioned interactions. For the interactions between customers and cloud managers, we adopt the options-based sequential auctions (OBSAs) to the cloud resource allocation paradigm. As compared to existing works, our framework can handle customers with heterogeneous demands, provide truthfulness as the dominant strategy, enjoy a simple winner determination, and preclude the delayed entrance issue. We also provide the performance analysis of the OBSAs, which is among the first in literature. For the interactions between cloud managers and CPs, we propose an auction-based scheme for resource gathering. Through incorporating the offered prices, we capture the heterogeneous desires of the CPs in leasing their resources. We conduct a comprehensive mathematical analysis of the two markets and identify the bidding strategy of the cloud managers.
\end{abstract}

}

\maketitle

\IEEEdisplaynontitleabstractindextext

%
\IEEEpeerreviewmaketitle

\ifCLASSOPTIONcompsoc
\IEEEraisesectionheading{\section{Introduction}\label{sec:introduction}}
\else
\section{Introduction}
\label{sec:introduction}
\fi

\vspace{-1mm}
\IEEEPARstart {M}{odern} society relies crucially on efficient processing of the massive amount of data collected from a variety of sources such as customers' information, wireless sensors, and statistical polls, for which cloud computing is a natural platform. Various cloud-based services are offered by different commercial companies such as Microsoft Azure~\cite{ref:azure}, Google Cloud~\cite{ref:google}, and Amazon EC2~\cite{ref:amazon}. Many companies are also anticipated to join this profitable market by offering
cloud services. The recent growth in the customers' demands has motivated the idea of sharing the resources in cloud networks~\cite{ref:interCloud}, where cloud owners can temporarily rent spare resources from one another to provide better services to the customers. It is anticipated that in the near future, large companies may dominate the entire cloud computing market by renting cloud resources from smaller or private companies. In that case, one of the most suitable candidates for modeling the corresponding cloud resource allocation may be the auction mechanism due to its simplicity, versatility, and a good match with the request and response paradigm in cloud networks. Recently, Amazon \textit{Spot Instances} is introduced as a simple auction-based framework for resource allocation, where users can bid for their requested cloud servers~\cite{ref:spotins}.
\subsection{Related works}\label{sub:related}

Auction theory provides a solid mathematical foundation for resource allocation among a set of resource-seeking customers and a set of resource providers. Hence, there exists a body of literature studying auction-based resource allocation in other contexts such as spectrum sharing in cognitive radio networks \cite{ref:cogRadio1,ref:cogRadio2,ref:cogRadio3,ref:cogRadio4,ref:cogRadio5}.

In modern cloud networks, cloud servers can be classified into different types according to their hardware and software configurations. Also, a bundle of cloud servers of different types may be required to meet the heterogeneous user demands simultaneously. Hence, earlier frameworks (e.g.,~\cite{ref:onetype}) that only consider one type of cloud servers and one type of tasks cannot well capture the reality of the market. On the other hand, cloud servers often switch between busy and idle states repeatedly and customers may join and leave the market at will. To capture this dynamism, it is more desirable to hold sequential auctions instead of a single-round auction. One simple approach is to hold a sequence of single-round auctions over time. However, as mentioned in \cite{ref:alggame}, single-round truthful auctions usually lose the truthfulness property when they are extended to sequential auctions. The truthfulness property ensures that customers cannot get higher rewards by manipulating their true valuations for the goods.
This consideration motivates us to go beyond the existing works on single-round auction (e.g., \cite{ref:nontruthfull1,ref:nontruthfull2,ref:oneround1,ref:oneround2}), and to seek truthful sequential auction solutions.

The most related works to ours are~\cite{ref:OneSide1,ref:onetypeOne,ref:bideveryround,ref:oneside2,DA}. In \cite{ref:OneSide1}, a novel bidding language is introduced based on categorizing the users into different groups with respect to their characteristics. Users are partitioned into three groups: job-oriented users, resource-aggressive users, and resource-aggressive users with time-invariant capacity requirements. A truthful online cloud auction mechanism is introduced on top of this bidding language. However, the original model only considers one type of the cloud servers. The authors have extended their proposed framework to the case with multiple types of tasks and servers in \cite{ref:onetypeOne}. However, the resulting model requires calculating a complex payment function for each arriving task and obtaining the allocation strategy by solving an optimization problem. These issues become a concern when handling real-time resource allocation in cloud networks with a large task arrival rate.

In order to model the multiple types of cloud servers and customers with heterogeneous demands, current literature has mainly focused on utilizing the combinatorial auctions for cloud resource allocation. Although combinatorial auctions can guarantee some favorable properties (such as truthfulness) in theory, it is well-known that determining the winner and its payment in combinatorial auctions is NP-hard, which renders them impractical in dynamic markets with real-time demands such as cloud networks. Also, these auctions are inherently designed for one-round selling. These issues of combinatorial auctions have promoted further research (e.g.\cite{ref:bideveryround,ref:oneside2}) on solving winner determination using simpler approximation methods or extending them to sequential combinatorial auctions. In \cite{ref:bideveryround}, the authors proposed a truthful mechanism for sequential combinatorial auctions. In this framework, besides complicated winner determination and payment identification process, when a user's task requires a bundle of cloud resources for more than one unit of time, the user has to bid in multiple rounds of auctions. This fact makes the framework inapplicable when users require uninterruptible processing of their tasks. In \cite{ref:oneside2}, the interactions between customers and cloud providers is modeled as an online combinatorial auction. The model of that work captures multiple types of cloud servers and heterogeneity of customers' demands. Also, it considers a sequential style of auction, in which winner determination is translated into a series
of one-round optimization problems. A truthful mechanism of selling is examined; and an approximate algorithm is proposed for one-round optimization. However, similar to \cite{ref:onetypeOne}, the need for solving multiple optimization problems using empirical methods in each round of the auction makes the framework complicated and computationally intensive. All of the aforementioned works and most of the contemporary literature focus on modeling the interactions between cloud managers and customers, whereas the resource gathering process for cloud managers is largely ignored. In the pioneering work~\cite{DA}, a general framework for intercloud networks is presented where the interactions between users and cloud providers is modeled by many-to-many auctions. Afterward, the interactions between cloud providers is modeled by a coalition game in which the cloud providers borrow resources from each other to fulfill their customers' demands. This work is among the first to consider the interactions both between the customers and cloud providers, and among the cloud providers themselves; unfortunately, users with bundle demands were not considered. Furthermore, one of the main challenges in dynamic cloud resource allocation scheme neglected in most of the mentioned works is the \textit{delayed entrance} problem. This problem arises when a user delays its entrance into the market when he has some side-information about the future dynamics of the market. Assume that in a sequential combinatorial auction, some users become aware that by waiting for some period of time, the cloud resources can be obtained at lower prices. In this scenario, all the users with side-information postpone their entrance into the market. In a large-scale market, this circumstance leads to a burst of arrival, and thus an unstable market. 

In summary, existing literature on modeling the interactions between the cloud managers and customers has at least one of the following four limitations: (i) incapability of handling customers' heterogeneous demands that require a bundle of different types of servers, (ii) missing the truthfulness property, (iii) requiring prohibitive computation for winner and payment determination, and (iv) susceptible to the delayed entrance issue. These issues will be addressed in our work. Also, to the best of our knowledge, our work is among the first to leverage auction theory to study the interactions among the public cloud managers and private cloud providers (CPs), which better captures the selfishness of the CPs considering their willingness to resource sharing captured by their \textit{offered prices}.

\subsection{Novelty and Contributions}
To address the limitations of the existing works mentioned in subsection~\ref{sub:related}, a novel two-stage auction framework is proposed in this work to capture the interactions among (a) customers and cloud managers, and (b) cloud managers and CPs. Specifically, we consider cloud of clouds networks $\mbox{\small (CCNs)}$ consisting of heterogeneous cloud servers and customers with different demands. There exists a $\mbox{\small CCN}$ manager in charge of handling the resources of each CCN. The CCN managers are interested in renting servers from CPs to enlarge their pool of resources so as to attract more customers and better handle their real-time demands. The first stage of the proposed framework is inspired by the options-based sequential auctions (OBSAs)~\cite{ref:OB} and models the interactions between customers and CCN managers, in which each customer endeavors to obtain his/her demanded resources from a CCN.\footnote{In this study, we consider the infrastructure as a service (IaaS) form of cloud computing.} To the best of our knowledge, we are (among) the first to leverage OBSAs to address the major limitations of existing works on dynamic cloud networks.\footnote{The truthfulness property is guaranteed in the second-price options-based sequential auctions.} In addition, we provided the corresponding performance analysis based on a novel Markov chain modeling, which is new to existing studies in the relevant literature. The second stage of the proposed framework describes the interactions between multiple CCN managers and multiple CPs, in which CCN managers compete to obtain resources from CPs.
For this stage, we introduce a novel model consisting of two parallel markets for gathering cloud resources: flat-price market and auction-based market, to better capture the selfishness of the private CPs (by incorporating offered prices) as compared to the existing models.\footnote{These markets may be viewed as the counterparts of the day-ahead and real-time markets in smart grids~\cite{powerDay}.} We also provide a comprehensive analysis for these markets using Hamilton-Jacobi-Bellman equation and derive the bidding strategy of the CCN managers with respect to their inherent characteristics in a stable market setting.

\textbf{Structure of the paper:} The system model is introduced in Section~\ref{sec:sysmodel}. The interactions between the customers and the CCN managers is modeled in Section~\ref{sec:interCCncust}. Section~\ref{sec:analysisOBSA} is devoted to the analysis of OBSAs. The interactions between the CPs and the CCN managers is modeled and analyzed in Section~\ref{sec:interacCCNandCPs}. Simulation results are presented in Section~\ref{sec:sim}. Finally, Section~\ref{sec:concl} concludes the paper and provides some possible future directions.
\section{System Model}~\label{sec:sysmodel}
\hspace{-2mm}A $\mbox{\small CCN}$ consists of multiple cloud servers with different processing capabilities; some of them are more desirable for GPU processing, while the others are more suitable for real-time database analysis and parallel processing. In addition to their core servers, CCNs can rent servers from CPs to process their backlogged tasks and to serve more customers. CPs are small cloud retailers who lease their extra computational resources to CCNs for profit.  Customers with multiple heterogeneous demands may join the CCN at will and require multiple types of servers simultaneously. Inspired by~\cite{ref:OB,DA}, proxy agents ($\mbox{\small PA}$) are incorporated into our model as trusted mediators between the customers and the corresponding CCN. Each customer sends its demands to an idle PA; subsequently, the PA attempts to fulfill the demands with the available resources of a CCN. Each $\mbox{\small CCN}$ operates under the control of a CCN manager who interacts with CPs and PAs.
\begin{figure}[t]
\includegraphics[width=3.4in, height=2in]{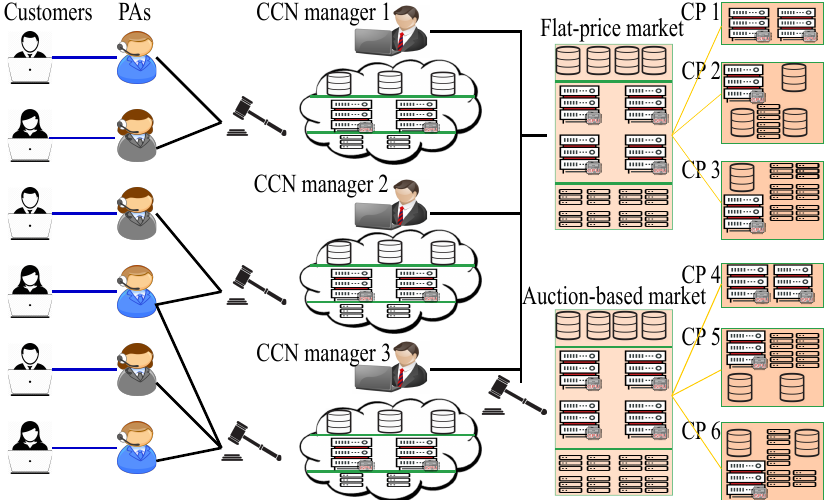}
\centering
\caption{Market model.}
\label{fig:1}
\end{figure}

Due to the variety in the task types and individual priorities, customers may have disparate preferences for different (combination of) servers, which is assumed known to their corresponding PAs. In this paradigm, the $\mbox{\small PAs}$ and the $\mbox{\small CCN}$ managers employ a common bidding language that reflects the customers' demands and valuations. Nevertheless, the discussion of the bidding language is beyond the scope of this paper. An interested reader is referred to~\cite{ref:18},~\cite{ref:22} and references therein for more details. 

In this work, we introduce a framework in which CCN managers rent extra servers from CPs by participating in one of the two parallel markets: the flat-price market and the auction-based market. In the flat-price market, the CPs offer their servers at a fixed price. In the auction-based market, CPs provide their servers along with their offered prices (i.e.,  the least expected price to lend the corresponding servers), where the CCN managers bid in a sequence of auctions to obtain the servers while satisfying the CPs' offered prices. The flat-price market is more suitable for leasing the servers with long idle periods. In this case, since a CP does not need to utilize the server in the near future, he aims to lease the server with a constant high price. However, the auction-based market is more favorable for servers with a short idle period. In this case, the CP may need its servers in a near future for itself. Hence, CPs compete with each other by offering lower prices of their servers so as to lease them faster. Similarly, a CCN manager who requires the resource immediately and needs to rent it for a long period tends to join the flat-price market, while the rest of the CCN managers participate in the auction-based market. As can be seen from Figure~\ref{fig:1}, the proposed model involves two stages for gathering and selling the resources. The first stage captures the interactions between the PAs and the CCN managers, while the second represents the interactions between the CCN managers and the CPs. In the following, we will introduce and analyze these two stages in order.

\section{Interactions between CCN managers and PAs: Options-based Sequential Auctions}\label{sec:interCCncust}
\noindent The main purpose of utilizing an auction is to sell goods when there is more than one interested buyers. In \textit{sequential auction}, the seller holds consecutive auctions for selling goods. Since the seller can adjust the time interval between the consecutive auctions, sequential auctions are suitable for the following scenarios: (i) availability of the goods varies over time, which means the goods may not be available in some of the time instances; (ii) the buyers arrive at the market at different times, which requires the seller to wait for some period of time before the number of buyers exceeds a threshold to guarantee a certain profit. Considering these facts, sequential auction is arguably the most suitable type of auctions for leasing the cloud servers to the PAs.

Classic first-price and second-price sequential auctions have been studied in the literature \cite{ref:seq1,ref:seq2}. Nevertheless, one of the main drawbacks of these auctions is the lack of a dominant strategy that can accommodate heterogeneous demands of buyers when customers face multiple sequential auctions. For example, consider the following two situations:  
\begin{enumerate}
\item A buyer with a limited budget and heterogeneous demands requires goods from either of the two sequential auctions but not from both.  
\item A buyer with a limited budget requires goods from both sequential auctions simultaneously.
\end{enumerate}
In both cases, a buyer has no dominant strategy for splitting the budget between multiple sequential auctions~\cite{ref:OB}. For the interactions between PAs and CCN managers, considering PAs as buyers, CCN managers as sellers, and each server as a good, holding a separate sequential auction for each type of servers will lead to the situations in the above-mentioned examples, which makes classic sequential auctions inapplicable. 
Another main concern of applying the classic sequential auctions and the combinatorial auctions to our context is the delayed entrance issue, where PAs may intentionally delay their entrance to the market for a better price. This phenomenon leads to an undesired burst of arrivals into the market and causes market instability. 

The aforementioned concerns suggests a need for an auction mechanism that can: (i) capture the dynamism of the market with its sequential style of holding; (ii) enjoy truthfulness as the dominant strategy to ensure that PAs cannot make more profit by manipulating their true valuations; and (iii) resolve the delayed entrance issue. To this end, we propose utilizing the options-based sequential auction (OBSA)~\cite{ref:OB} in our study. Besides enjoying all the above-motioned characteristics, OBSA leads to more trust between bidders and an auctioneer, and thus results in a higher long-term profit for the auctioneer. Moreover, OBSAs have simple implementation and admit fast winner recognition.
 
 Basically, OBSAs are classic sequential auctions reinforced with the options-based property. Consider a classic sequential auction with first- or second-price backbone in our context. In the first- (second-)price scenario, the PA with the highest bid is the winner who is then required to make a payment equal to his (the second highest) bid. Considering OBSAs in our context, the options-based property guarantees the least payment for winner PAs during their patience time, where the patience time is referred to as the time window in which the PA can wait before utilizing its obtained resources. Hence, this property eliminates the sensitivity of the PAs' payment to the time of winning an auction. In OBSAs, winner PAs are granted the opportunity to collect all of their demanded resources from the CCN resource pool before getting charged. The options-based property manifests itself through the price matching process, which is the main difference between OBSAs and classic sequential auctions. In~\cite{ref:OB}, OBSAs with second-price backbone are proposed without mathematical analysis. In this work, we introduce them to the cloud-related literature, adapt them to the cloud resource allocation scenario, and provide mathematical analysis identifying multiple performance metrics of interest. We also present the OBSAs with the first-price backbone, which builds the foundation of analysis for the OBSAs with the second-price backbone.

 \RestyleAlgo{boxruled}
 \begin{algorithm}[t]
 	\footnotesize
 	\caption{Price matching process for the first-price \newline options-based sequential auction}\label{alg:1}
 	\SetKwFunction{Union}{Union}\SetKwFunction{FindCompress}{FindCompress}
 	\SetKwInOut{Input}{input}\SetKwInOut{Output}{output}
 	\Input{Current price that the $\mbox{\small PA}$ has to pay when the current auction begins, entering and patience time of the winner $\mbox{\small PA}$ ($P_{cur}$,$t_{ent}$,$t_{pat}$), clock time ($T$), winner of the current auction's bid ($b_w$)}
 	\Output{Price that the $\mbox{\small PA}$ has to pay ($P_{out}(T)$)}
 	\BlankLine
 	$t \longleftarrow t_{ent}+t_{pat}$\\
 	$P_{out}(T) \longleftarrow P_{cur}$\\
 	\If{($t\geq T$ and $b_w<P_{cur}$)}{
 		$P_{cur} \longleftarrow b_w$\\
 	}
 	$P_{out}(T) \longleftarrow P_{cur}$\\
 \end{algorithm}\DecMargin{1em}
 \RestyleAlgo{boxruled}
 \begin{algorithm}[h]
 	\footnotesize
 	\caption{Price matching process for the second-price \newline options-based sequential auction}\label{alg:2}
 	\SetKwFunction{Union}{Union}\SetKwFunction{FindCompress}{FindCompress}
 	\SetKwInOut{Input}{input}\SetKwInOut{Output}{output}
 	\Input{Current price that the $\mbox{\small PA}$ has to pay when the current auction begins, entering and patience time of the winner $\mbox{\small PA}$ ($P_{cur}$,$t_{ent}$,$t_{pat}$), stored bumped $\mbox{\small PA}$'s ID in this $\mbox{\small PA}$'s memory ($ID_{mem}$), clock time ($T$), bid and ID of the current auction's winner ($b_w$,$ID_{w}$), bid and ID of the current auction's bumped PA ($b_{bumped}$,$ID_{bumped}$)}
 	\Output{Price that the $\mbox{\small PA}$ has to pay ($P_{out}(T)$)}
 	\BlankLine
 	$t \longleftarrow t_{ent}+t_{pat}$\\
 	$P_{out}(T) \longleftarrow P_{cur}$\\
 	\If{($t\geq T$)}{
 		\eIf{($ID_{mem} \neq Null$)}{
 			\eIf{$ID_{w} = ID_{mem}$ }{
 				\If{$b_{bumped}<P_{cur}$}{
 					$ID_{mem} \longleftarrow ID_{bumped}$\\
 					$P_{cur} \longleftarrow b_{bumped}$\\
 				}
 			}
 			{
 				\If{$b_w<P_{cur}$}{
 					$ID_{mem} \longleftarrow Null$\\
 					$P_{cur} \longleftarrow b_{w}$
 				}
 			}
 		}
 		{
 			\If{$b_{w}<P_{cur}$}{
 				$P_{cur} \longleftarrow b_w$
 			}
 		}
 	}
 	$P_{out}(T) \longleftarrow P_{cur}$\\
 \end{algorithm}\DecMargin{1em}

\subsection{First- and Second-Price OBSAs and PA's Role}\label{subsec:1}
  Without loss of generality, we investigate the interactions between multiple PAs and a CCN manager (for one CCN) in the following. Note that we consider sealed bidding procedure, in which $\mbox{\small PAs}$ do not communicate with each other and have no information about each others' bid. A CCN manager holds OBSAs with either the first-price or the second-price backbones for its possessed cloud servers, where a separate sequence of auctions is offered for each type of servers. The first- and second-price OBSAs and the PA's corresponding actions are discussed in the following:   

\textbf{1) Calculating the bid for each of the auctions of interest:} The $\mbox{\small PA}$ submits a bid equal to the customer's \textit{maximum marginal value} for suitable auctions~\cite{ref:OB}. By pursuing this bidding strategy, the $\mbox{\small PA}$ wins those types of servers which are potentially profitable for the corresponding customer.

\textbf{2) Obtaining the best price (price matching):} There are two modes of operation for any $\mbox{\small PA}$ in an auction: participant mode and observer mode. When a $\mbox{\small PA}$ enters a $\mbox{\small CCN}$, he becomes a participant and participates in the appropriate active auctions. From the moment that the $\mbox{\small PA}$ wins an auction, he switches to the observer mode for that auction. In this mode, the observer $\mbox{\small PA}$ reduces the price of a won server to a lower price using the following procedure:\footnote{It is assumed that an observer PA knows the identity and the bid of the winner PA and those of the PA with the second highest bid in the observed auction.} 

\textbf{A) OBSAs with the first-price backbone:} The observer $\mbox{\small PA}$ decreases his current payment to the winner's bid if the winner wins the auction with a lower price as compared to the PA's current payment. Otherwise, the observer PA does not make any changes to his current payment. The price matching process for the first-price OBSA is summarized in Algorithm~\ref{alg:1}.

\textbf{B) OBSAs with the second-price backbone:} In this case, each $\mbox{\small PA}$ has an identity that gets updated whenever he enters the market on behalf of a new customer. The winner $\mbox{\small PA}$ stores the identity of the PA who has proposed the second highest bid, i.e., the so-called bumped $\mbox{\small PA}$. Possible situations for the subsequent auction and the corresponding actions of the observer $\mbox{\small PA's}$ are as follows:

\textbf{B-1)  The bumped $\mbox{\small PA}$ wins the next auction:} The observer $\mbox{\small PA}$ decreases his current payment to the second highest bid of the next auction and updates his memory by saving the identity of the PA who has proposed this bid.\footnote{This is due to the fact that if the observer $\mbox{\small PA}$ would have postponed his entrance up to this auction, he would have to pay this price for this server.}

\textbf{B-2)  The bumped $\mbox{\small PA}$ stays at the market but loses the next auction:} This implies that the winner's bid is higher than the current payment of the observer $\mbox{\small PA}$. This is due to the fact that the true valuation of the bumped $\mbox{\small PA}$ is time-invariant and that truthfulness is a dominant strategy in the second-price OBSAs~\cite{ref:OB}. In this case, the observer $\mbox{\small PA}$ will neither change its memory nor its current payment.

\textbf{B-3) The bumped $\mbox{\small PA}$ leaves the market:} The observer $\mbox{\small PA}$ clears his memory and, from that point, he decreases his current payment to each of the successive winner's bid if it is lower than its current payment (same as the first-price backbone).\footnote{This is due to the fact that in this case, he could not have changed the market with delaying his entrance, and he would have to pay this price if he entered the market at the beginning of this auction.}

The price matching process for the second-price OBSAs is summarized in Algorithm~\ref{alg:2}.

\textbf{3) Utilizing the won servers:} At the end of the PA's patience time, among his obtained (won) servers, he chooses those maximizing the corresponding customer's utility. Mathematically, the $\mbox{\small PA}$ chooses a set of servers $s^*$ obtained as $s^* = \argmax \limits_{s \subseteq S}[v(s)-p(s)]$, where $S$ denotes the won servers by the $\mbox{\small PA}$ during his patience time, and $p(s)$, $v(s)$ denote the payment and the customer's valuation corresponding to the servers belonging to set $s$, respectively. All the other acquired servers by this $\mbox{\small PA}$ will be returned to the $\mbox{\small CNN}$ without any charge.

As can be observed, in OBSA, the winner determination is very simple in both first- and second-price backbones since the winner for each auction is always the PA with the highest bid. Also, payment calculation is fairly simple through observing the current winner bid. In contrast, combinatorial auction, proposed in contemporary cloud-related literature (e.g.,~\cite{ref:bideveryround,ref:oneside2}), requires solving the complex winner determination process and calculating complex payment functions for each PA. 

\vspace{-2mm}  
\section{Analysis of OBSAs}~\label{sec:analysisOBSA}
\vspace{-7mm} 
\subsection{Backgrounds}
\noindent The existence of the price matching process makes the analysis of OBSAs completely different from that of classic sequential auctions. In this section, mathematical analysis of OBSAs is presented for both first-price and second-price backbones. Without loss of generality, the analysis is performed for an OBSA for one type of servers. It is assumed that the $i^{th}$ PA's bid ($b_i$) takes a discrete value such that $b_i \in \left[v_{min},v_{max}\right]$, $\forall i$, where $v_{max}-v_{min}=K\delta$, $K\in \mathbb{Z}^+$, and $\delta$ is the quantization size of the bids. In this case, $v_{min}$ is the $\mbox{\small CCN}$ manager's lowest affordable price (corresponding to zero profit) and $v_{max}$ could be an upper bound obtained through market survey. Also, PAs' bids are assumed to be independent and identically distributed (i.i.d.) in each auction. It is further assumed that the average number of $\mbox{\small PAs}$ in the market follows a Poisson distribution with mean $\Phi$, where, at each round of auction, $\Omega$ portion of them participate in the auction. Hence, the number of $\mbox{\small PAs}$ for each auction follows a Poisson distribution with mean $\lambda=\Omega\Phi$.\footnote{This implies that the probability of having $k$ $\mbox{\small PAs}$ in an auction is $p(\# \mbox{\small PA}=k)=\frac{e^{-\lambda} \lambda^k}{k!}$.} Our following analysis hold for an arbitrary distribution of the PAs' bids; however, specific bid distributions are needed to obtain more concrete results for further insights.
For this purpose, two specific distributions for PAs' bids that are representative for a wide variety of markets  are considered below: 

\textbf{1) Uniform distribution:} Bids are uniformly distributed between $\left[v_{min},v_{max}\right]$ such that $\frac{v_{max}-v_{min}}{\delta}=  K_u$, where $\delta$ is the quantization size of the bids. In this case, the length of bid interval is $v_{max}-v_{min}= \delta K_u$.

\textbf{2) Sampled Laplace distribution:} Bids are sampled from a continuous Laplace distribution $Laplace(\mu,w)$, where the $\mu$ and $w$ are the so-called location and scale parameters, respectively. In this case, the original continuous probability density function (pdf) is given by $f(x|\mu,w)=\frac{1}{2 w} e^{-\frac{|x-\mu|}{w}}$. This continuous pdf is discretized with the step size $\delta$ such that $v_{min} = \mu - K_{L_1} \delta$ and $v_{max}= \mu + K_{L_2} \delta$, where $K_{L_1}, K_{L_2}$ are positive integers. In this case, the length of bid interval is $v_{max}-v_{min} = (K_{L_1}+K_{L_2}) \delta$.

We propose the sampled Laplace distribution to analyze a market in which the bids are concentrated around a specific value and become less probable elsewhere. We formally define this distribution and provide some related properties below.
\begin{definition}
A random variable $Y$ is defined as a \textbf{concatenated sampled (discretized) Laplace random variable} on the interval $[a,b]$ with the sampling step size of $\delta \in \mathbb{Z}>1$, and with the location and scale parameters $(\mu,w)$ if and only if the following conditions are satisfied:
\begin{itemize}
\item $\frac{b-a}{\delta}= K $, where $K \in \mathbb{Z}>1$.
\item $\mu = a + \beta \delta$, where $\beta \in \mathbb{Z}>1$.

\item $f_Y(a+k' \delta)=p_Y(Y=a+k' \delta)=\frac{1}{\Gamma(\delta,K,\beta,w)}\big(\frac{1}{2 w} e^{-\frac{|a+k' \delta-\mu|}{w}}\big)$, where $k' \in \mathbb{Z}, 0\leq k' \leq K$, and $\Gamma(\delta,K,\beta,w)$ is a normalization factor.
\end{itemize}
\end{definition}
\begin{lemma}\label{lem:1}
The normalization parameter $\Gamma(\delta,K,\beta,w)$ of a concatenated sampled (discretized) Laplace random variable with specified parameters is given by:
\begin{gather}
\Scale[0.99]{\Gamma(\delta,K,\beta,w) = 
	\frac{1}{2 w} e^{-\frac{ \delta \beta}{w}} \frac{1-(e^{\frac{\delta \beta}{w}})}{1- e^{\frac{\delta}{w}}} }\nonumber\\   
\Scale[0.99]{+\frac{1}{2 w} e^{-\frac{\delta}{w}} \frac{1-(e^{-\frac{\delta(K-\beta)}{w}})}{1- e^{-\frac{\delta}{w}}}+ \frac{1}{2w} .}
\end{gather}
\end{lemma}
\begin{proof}
The proof can be carried out based on the fact that sum of the $f_Y(x)$ over all possible values of $x$ is $1$, i.e., $\sum_{k'=0}^{K}f_Y(a+k' \delta)=1$.
\end{proof}
\begin{fact}
The cumulative distribution  function (cdf) of a concatenated sampled (discretized) Laplace random variable on the interval $[a,b]$ with the sampling step size of $\delta \in \mathbb{Z}>1$, and with the location and scale parameters $(\mu,w)$ is given by:

	\begin{align*}\label{Eq:FL23}
	F_Y&(a+ k\delta)= p_Y(y\leq a+ k\delta) = \nonumber  \\ 
		\end{align*}
		\begin{align}\label{Eq:FL}
&	\begin{cases}
	0, \hspace{40mm} \textrm{if} \;k<0, \\
	\frac{1}{2w \Gamma(\delta,K,\beta,w)} e^{\frac{a-\mu}{w}} \frac{1-e^{\frac{\delta (k+1)}{w}}}{1-e^{\frac{\delta}{w}}}, \;\; \textrm{if} \;0\leq k<\beta,\\ 
	\frac{1}{2w \Gamma(\delta,K,\beta,w)} \big( e^{\frac{a-\mu}{w}} \frac{1-e^{\frac{\delta \beta}{w}}}{1-e^{\frac{\delta}{w}}} + 1 \big), \; \textrm{if} \; k=\beta,\\
	\frac{1}{2w \Gamma(\delta,K,\beta,w)}\big( e^{\frac{a-\mu}{w}} \frac{1-e^{\frac{\delta \beta}{w}}}{1-e^{\frac{\delta}{w}}} + 1 \\
	+ e^{-\frac{\delta}{w}}\frac{1-e^{-\frac{\delta (k-\beta)}{w}}}{1-e^{-\frac{\delta}{w}}}  \big ),  \; \hspace{13mm}\textrm{if} \; K \geq k>\beta, \\
	1, \hspace{39.5mm} \textrm{if} \; k> K.
	\end{cases}
	\end{align}
\end{fact}	
\begin{proof}
The cdf can be defined as follows:
\vspace{-2mm}
	\begin{align}
p_Y(Y\leq a+ k\delta) = \sum_{k'=0}^{k} p_Y(Y=a+ k'\delta)\nonumber  \\ 
= \frac{1}{\Gamma(\delta,K,\beta,w)} \sum_{k'=0}^{k} {\frac{1}{2 w} e^{-\frac{|a+k'\delta-\mu|}{w}}}.
\end{align}
Applying the result of Lemma~\ref{lem:1} and some mathematical manipulations give us the results above.
\end{proof}
\textbf{Remark:} For a better illustration, some exemplary concatenated sampled Laplace distributions are depicted in Figure~\ref{diag:bids_dist2}.
\vspace{-1mm}
\begin{figure}[h]
	\includegraphics[width=3.5in, height=2in]{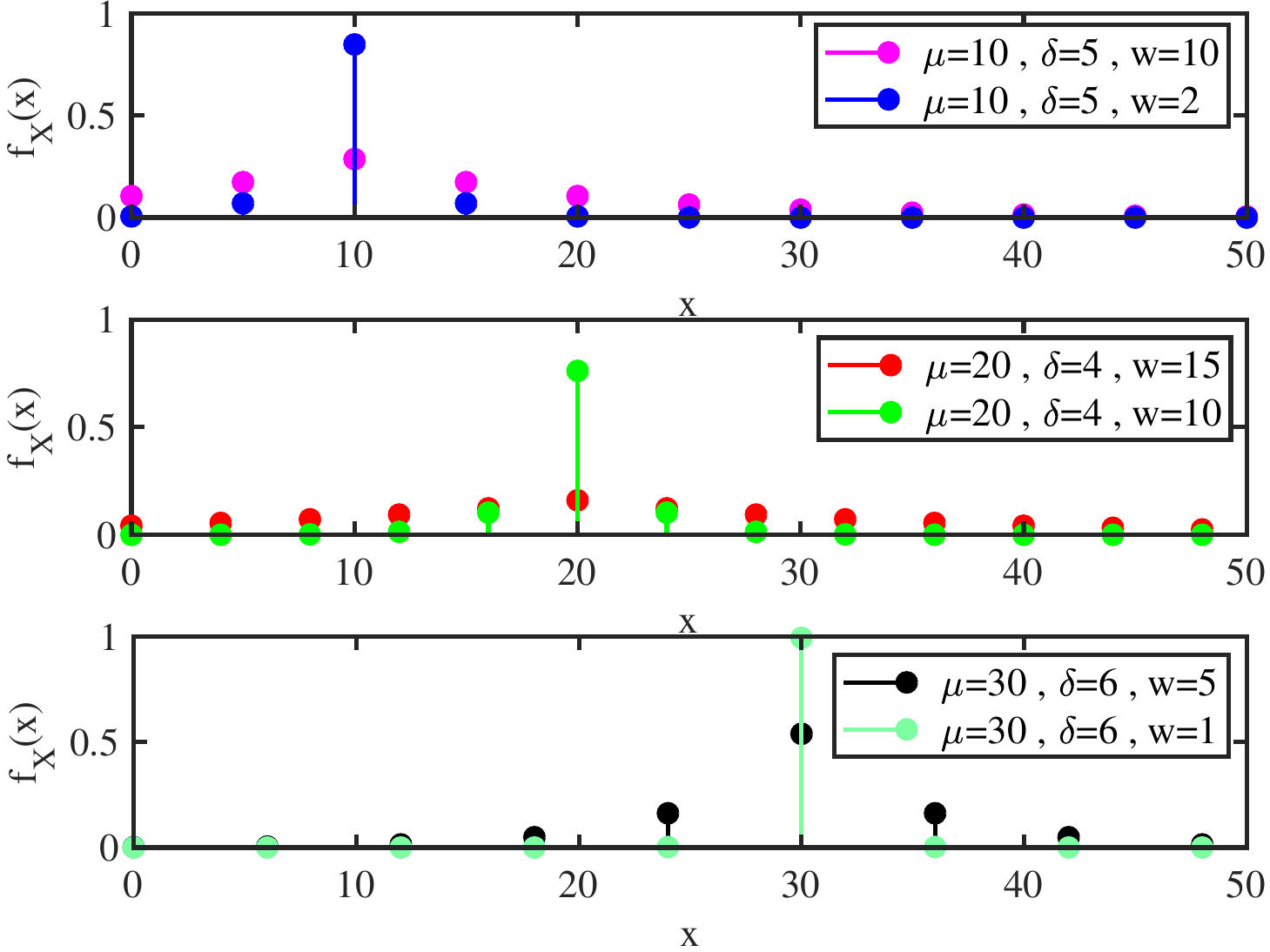}
	\centering
	\caption{Concatenated sampled Laplace distributions on the interval [0,50] with different parameters.}
	\label{diag:bids_dist2}
\end{figure}

\vspace{-1mm}
In the rest of this section, systematic analysis of the price matching process of OBSAs is examined. We start by presenting the analysis of the first-price OBSAs. Afterward, the extension to the second-price OBSAs is presented.
\vspace{-0.3mm}
\subsection{First-price OBSAs}
\vspace{-0.1mm}
\subsubsection{Modeling the price matching process}
We model the price matching process of an observer $\mbox{\small PA}$ by a discrete time homogeneous Markov chain ($\mbox{\small DTHMC}$), where the state space consists of the possible payments of the PA. Initially, a winner $\mbox{\small PA}$ enters into one of the states depending on his current payment. Afterward, this PA transits among different states of the DTHMC during his observation time, where his current state always represents his current payment during the price matching process. At the end of his observation time window, he leaves the DTHMC and makes a payment according to his last state. To carry out the analysis, we define the \textit{residual patience time} $\Delta$ as the time duration in which the observer $\mbox{\small PA}$ stays in the proposed DTHMC (i.e., the length of the observation mode).

Figure~\ref{fig:5} depicts the transition diagram of the proposed DTHMC, where $L=(v_{max}-v_{min})/ \delta$, $p_{ij}$ denotes the transition probability from state $i$ to state $j$, and $\theta_{j}$ denotes the \textit{stopping (leaving) probability} for state $j$. The leaving probabilities indicate quitting the price matching process upon having zero residual patience time and making a payment with respect to the last state. Also, $\pi^0_{j}$ denotes the \textit{entering probability} for state $j$. This parameter is defined based on the probability of entering the DTHMC states for a winner PA which depends on the initial bids. Considering a lower bound on the bids reported by the $\mbox{\small CCN}$ manager, all the states  of the DTHMC are transient except for the last state. Also, due to the price matching process, only transitions to lower prices (states with smaller indices) are possible.
\vspace{-4mm}
\begin{figure}[h]
\includegraphics[width=0.35\textwidth]{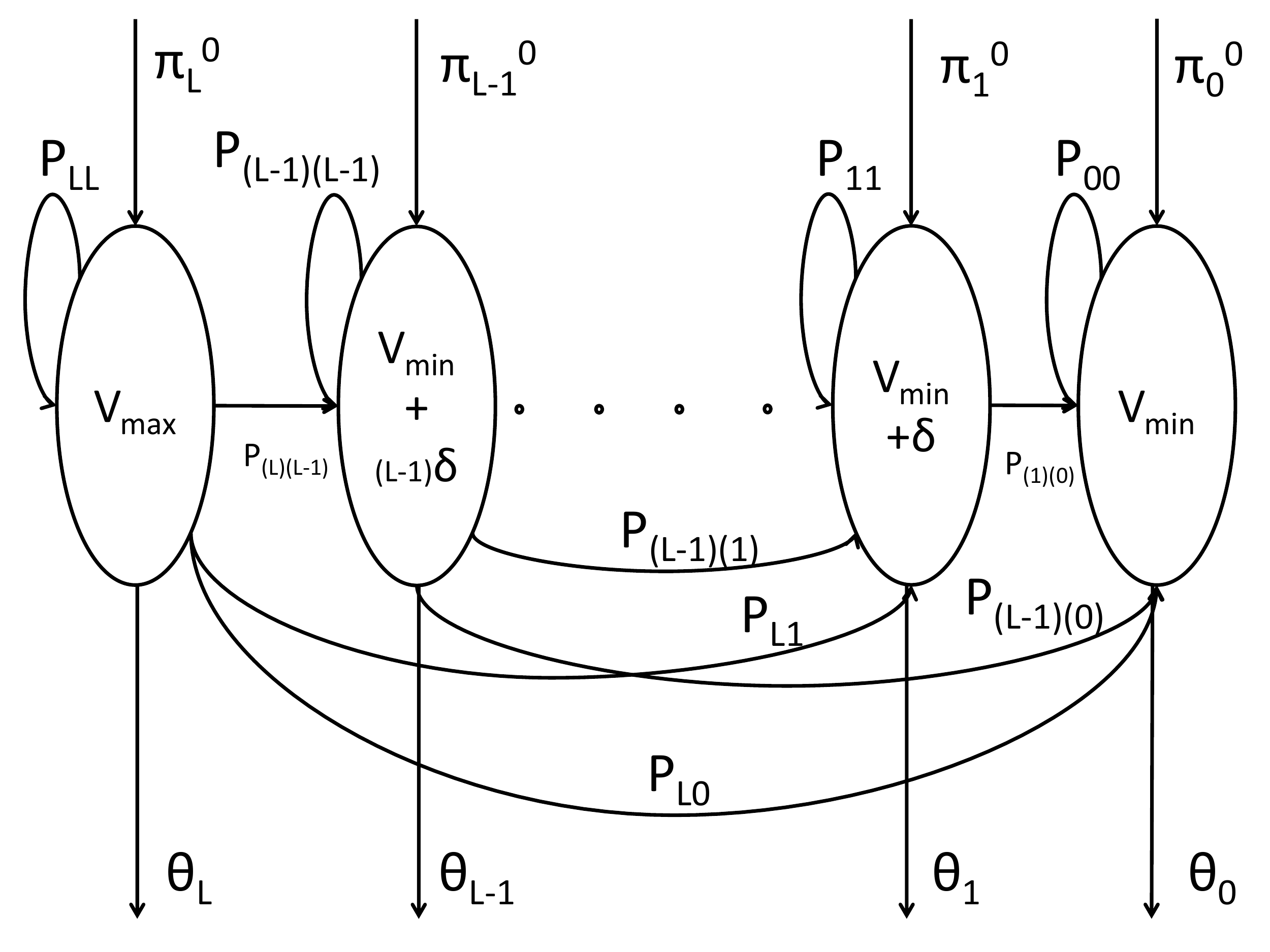}
\centering
\caption{Transition diagram for the price matching process in the first-price OBSA.}
\label{fig:5}
\end{figure}

\vspace{-2mm}
In the first-price OBSAs, for an observer $\mbox{\small PA}$ with current state (current payment) $x$, a transition to another state only occurs when another $\mbox{\small PA}$ wins the observed auction with a bid less than $x$. Hence, transition probabilities for the modeled $\mbox{\small DTHMC}$ can be derived as follows:\footnote{Note that all the derived transition probabilities and entering probabilities depend on the number of available PAs in the market which is a Poisson random variable as defined before.}
\begin{gather}\label{eq:p}
\Scale[0.90]{p_{xy}= \sum_{N=1}^{\infty} p\left(max(b_1,\cdots ,b_N)=v_{min}+y \delta \right) p(\# \mbox{\small PA}=N) } \nonumber\\  
\;\; \Scale[0.95]{,  L \geq x>y\geq 0 ,}\\ 
\Scale[0.90]{p_{kk}=\sum_{N=1}^{\infty}p\left(max(b_1,\cdots ,b_N)\geq v_{min}+k \delta \right) p(\# \mbox{\small PA}=N) }\nonumber\\  
 \;\;\Scale[0.95]{, L \geq k\geq 0,} 
\end{gather}
where $N$ is the number of participant $\mbox{\small PAs}$ in the auction. The following expressions for the transition probabilities can be derived readily.
\begin{itemize}
\item \textbf{Uniform bids}:
\begin{equation}\label{eq:cite1}
\begin{aligned}
p_{xy}&=\sum_{N=1}^{\infty}  \frac{(y+1)^N-y^N}{(L+1)^N} \left(\frac{e^{- \lambda} \lambda ^N}{N!}\right)\\
&= e^{- \lambda}\left[e^{\frac{(y+1)\lambda}{L+1}}-e^{\frac{y \lambda}{L+1}}\right].
\end{aligned}
\end{equation}
\item \textbf{Concatenated sampled Laplace bids}: 

\begin{align}\label{eq:olgooo}
p_{xy}&= \sum_{N=1}^{\infty} \frac{e^{- \lambda} \lambda ^N}{N!}\left(\frac{1}{2w \Gamma(\delta,K,\beta,w)}\right)^N \nonumber \\
&\Bigg[ \left(\sum_{k'=0}^{y} {e^{-\frac{|v_{min}+k'\delta-\mu|}{w}}}\right)^N - \nonumber \\
& \left(\sum_{k'=0}^{y-1} {e^{-\frac{|v_{min}+k'\delta-\mu|}{w}}}\right)^N \Bigg]\nonumber \\
 &= e^{-\lambda}[e^{\frac{\lambda F_Y(v_{min}+y\delta)}{2w \Gamma(\delta,K,\beta,w)}}-e^{\frac{\lambda F_Y(v_{min}+(y-1)\delta)}{2w \Gamma(\delta,K,\beta,w)}}].
\end{align}
This expression can be further expanded using~\eqref{Eq:FL}.
\end{itemize}

The winner PA enters a state depending on his initial bid. Therefore, the entering probabilities for the $\mbox{\small DTHMC}$ can be defined as:
\begin{gather}
\pi _k^0=p\left(max(b_1,\cdots ,b_N)= v_{min}+k\delta \right)p(\# \mbox{\small PA}=N)\\ \nonumber
 =p_{xk},  
 \;\; L \geq x > k\geq 0.
\end{gather}

Using the transition probabilities, the \textit{one-step transition matrix}
$P=[p_{ij}]_{0\leq i,j\leq L}$, which fully describes the characteristics of the proposed DTHMC, can be readily built. Note that $P$ is a lower triangular matrix since the price
matching process only allows transitions toward lower states. 
\subsubsection{Expected income of the $\mbox{\small CCN}$ manager in the first-price OBSAs}\label{subsec:fp}

\noindent One of the important performance metrics of any auction is the auctioneer's income or profit. The $\mbox{\small CCN}$ manager's expected income is the same as the $\mbox{\small PAs}$' expected payments. Due to the construction mechanism of our proposed DTHMC, calculation of the leaving (stopping) probabilities is the prerequisite for obtaining the expected $\mbox{\small PA}$'s payment.
\begin{proposition} \label{th:1} For an arbitrary observer $\mbox{\small PA}$ with an arbitrary residual patience time $\Delta$, the conditional leaving probabilities $\theta^i _j=p( \textrm{Ending at $j^{th}$ state given starting from the $i^{th}$ state})$ can be obtained as follows:
\vspace{-3mm}
\begin{gather}\label{eq:PE}
  \theta^i _j=(P^{\Delta})_{ij}=  \bigg[ \sum_{k=1}^{\Delta}p_{ij}p_{jj}^{\Delta - k}p_{ii}^{k-1}\\ \nonumber
 + \sum_{k=0}^{\Delta-2} \sum_{c=0}^{\Delta-2-k} \sum_{r=j+1}^{i-1} p_{ir}p_{rj}p_{jj}^{\Delta-k-2-c}p_{i,i}^k p_{r,r}^c  \nonumber
 + \cdots \bigg], \nonumber
\end{gather}
where $P^{\Delta}$ denotes the matrix $P$ to the power of $\Delta$.
\end{proposition}
\begin{proof}
The result of this proposition is based on the characteristics of the one step transition matrix of a Markov chain. Note that each of the terms in \eqref{eq:PE} corresponds to a different scenario of transitions before leaving the DTHMC, and it can be derived directly by inspecting Figure~\ref{fig:5}.
\end{proof}
The following corollaries follow readily from the above proposition.
\begin{corollary}
Considering the scenario described in Proposition~\ref{th:1}, the $\mbox{\small CCN}$ manager's expected revenue is given by:
\begin{equation}\label{eq:rev}
 E[\textrm{Revenue}]=\sum_{i=0}^{L}\sum_{j=0}^{i} j\theta^i_j \pi^0_i=\sum_{i=1}^{L} \sum_{j=1}^{i} j \left(P^{\Delta}\right)_{ij}  \pi^0_{i}.
\end{equation}
When calculating the expected income, $j$ should be replaced by $j+v_{min}$ in the above equation. 
\end{corollary}
To avoid the burdensome of matrix multiplications needed in the above corollary, lower bounds of the CCN manager's income can be used as explained in the following corollary. 
\begin{corollary} 
For the described scenario in Proposition~\ref{th:1}, the following expression can be obtained as a lower bound of the  $\mbox{\small CCN}$ manager's expected revenue:\
\vspace{-4mm}
\begin{gather}\label{eq:rev1}
\hspace{-3mm}\Scale[0.86]{ E[\textrm{Revenue}]\geq \sum_{j=0}^{L} \pi^0_j p_{jj}^{\Delta} +  \sum_{i=0}^{L} \sum_{j=0}^{i} j p_{ij} p_{jj}^{\Delta-1} \frac{1-\left(\frac{p_{ii}}{p_{jj}}\right)^{\Delta}}{1-\left(\frac{P_{ii}}{p_{jj}}\right)} \pi _{i}^{0}.}
\end{gather}
Also, obtaining better lower bounds by involving more higher-order terms is straightforward.
\end{corollary}
\subsection{Second-price OBSAs}
\subsubsection{Modeling the price matching process}
Considering the price matching process of second-price OBSAs presented in subsection~\ref{subsec:1}, it can be seen that as long as the bumped $\mbox{\small PA}$ stays in the market, the observer $\mbox{\small PA}$ will only change his current payment when the bumped $\mbox{\small PA}$ wins one of the subsequent auctions. This is because if the bumped $\mbox{\small PA}$ stays in the market and does not win the subsequent auctions, the winner PA's bid will be higher than the observer $\mbox{\small PA}$'s current payment. Nevertheless, when the bumped $\mbox{\small PA}$ leaves the market, the observer $\mbox{\small PA}$ clears his memory and adapts his payment to the winner PA's bid at each of the subsequent auctions.

Similar to the first-price OBSAs, Figure~\ref{fig:6} depicts the transition diagram of our proposed DTHMC corresponding to the actions of an observer PA in the second-price OBSAs. The transition diagram is composed of two main components: a primary Markov chain and multiple sub-Markov chains. The primary Markov chain describes the case in which the bumped PA stays in the market, while the sub-Markov chains capture the actions upon leaving the bumped PA. Each $sub^n_m$ in Figure~\ref{fig:6} refers to the case of leaving from the $m^{th}$ state of the primary Markov chain and entering a sub-Markov chain as depicted in Figure~\ref{fig:5} with an initial state $n$ ($\pi^0_n=1$), with $n\leq m$. This corresponds to the case in which the current payment of the observer PA is $m$ and in the next auction the bumped PA leaves the market and the observer PA's current payment becomes $n$ with regard to the winner's PA bid in the subsequent auction. Each vector $\underline{Z_{m}}=[z_{m0},z_{m1},z_{mm},0, \cdots,0]$ denotes a ($1\times (L+1)$) row vector, where the $z_{mj}$ is the probability of leaving the primary Markov chain at state $m$ and entering the corresponding sub-Markov chain with starting state $j$ ($\pi^0_j=1$). Note that each departure branch denoted by $\underline{Z}_m, 0\leq m \leq L$, accounts for the $m+1$ possible transitions to the sub-Markov chains, which are grouped in Figure~\ref{fig:6} in the interest of space. In this diagram, $q_{ij}$ denotes the transition probability between state $i$ and $j$, which occurs when the bumped $\mbox{\small PA}$ stays in the market and the observer $\mbox{\small PA}$ reduces its payment from ($v_{min}+i\delta$) to the second highest bid in the subsequent auction ($v_{min}+j\delta$) (if it is less than its current payment).\footnote{In this case, to avoid confusion with the first-price OBSAs, the one-step transition matrix for the primary Markov chain is denoted by $Q$.} It is assumed that the bumped $\mbox{\small PA}$ participates in the next auction with probability $p(B)$, which is modeled as a complementary cdf of a memoryless distribution (e.g., exponential distribution) for the ease of analysis.\footnote{The memoryless property is considered so as to have a constant probability of leaving for the bumped PA between consecutive auctions.}
\begin{figure}[h]
\includegraphics[width=0.35\textwidth]{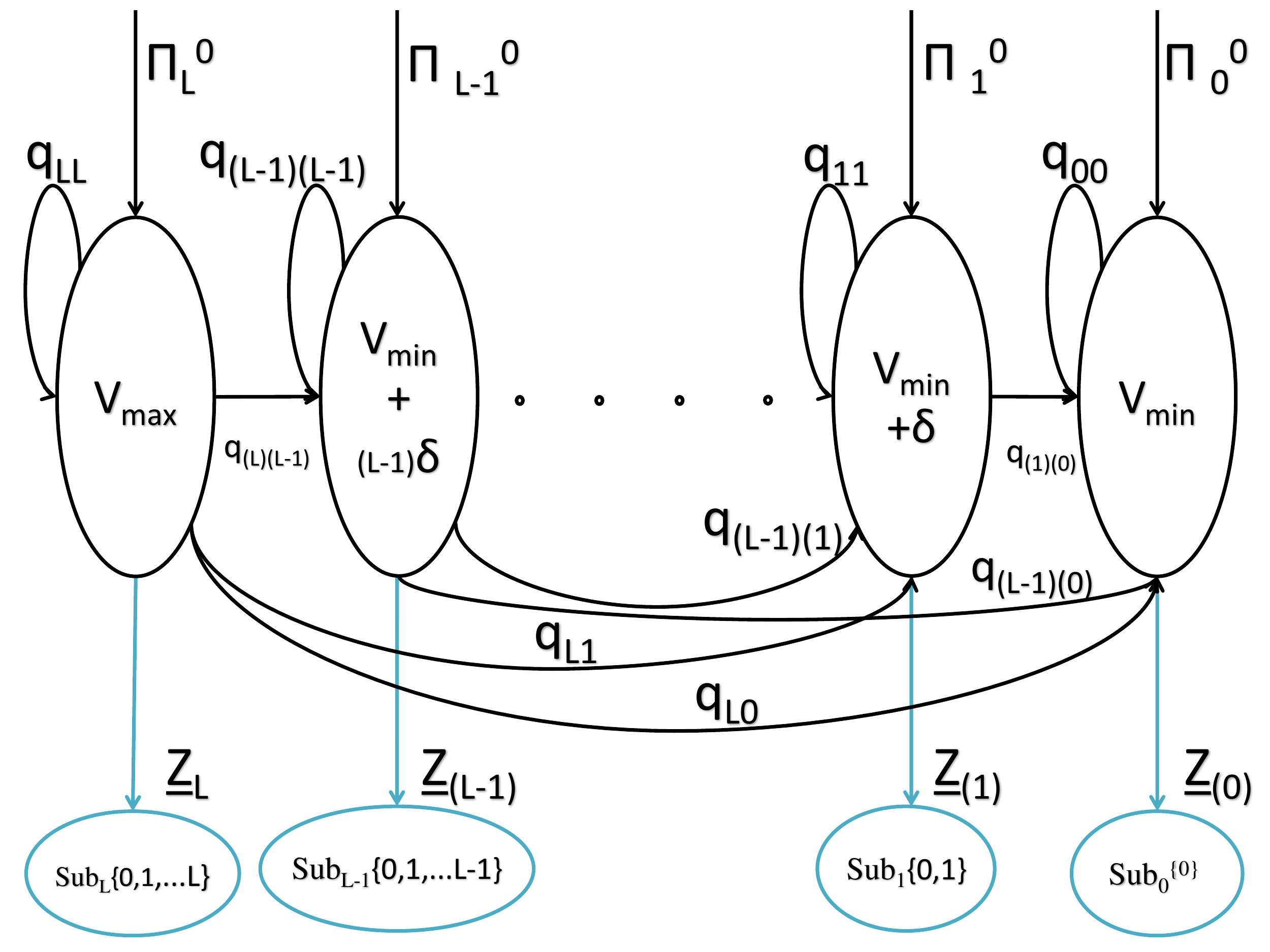}
\centering
\caption{Transition diagram for the price matching process in the second-price OBSAs.}
\label{fig:6}
	\end{figure}
According to the same logic used for deriving the transition probabilities in the first-price OBSAs, we can obtain the elements of vector $\underline{Z_{x}}$ for state $x$ as follows: 
\begin{gather}\label{eq:14}
\Scale[0.90]{z_{xy}= \sum_{N=1}^{\infty}\big(1-p(B)\big) p\left(\max (b_1, \cdots , b_{N})= y \delta +v_{min} \right)}\nonumber \\ 
\Scale[0.90]{	 p(\# \mbox{\small PA}=N), \;\; L \geq x > y \geq 0,}\\ 
\Scale[0.90]{z_{xx}= \sum_{N=1}^{\infty}\big(1-p(B)\big) p\left(\max (b_1, \cdots , b_{N}) \geq x \delta +v_{min} \right)} \nonumber\\ 
\Scale[0.90]{ p(\# \mbox{\small PA}=N), \;\; L \geq x \geq 0.} 
\end{gather}
 Note that $z_{xx}$ corresponds to the situation in which the bumped $\mbox{\small PA}$ leaves the market but the winner PA's bid at the next auction is higher than the observer $\mbox{\small PA}$'s current payment.
 
 For the primary Markov chain, a transition between states only occurs when the bumped $\mbox{\small PA}$ stays in the market. Hence, the elements of the transition matrix $Q=[q_{xy}]_{0\leq x,y\leq L}$ can be derived as follows:\footnote{The result can be found by changing parameters in the derived formula for $z$.}
	\begin{gather} 
q_{xy}=  \sum_{N=1}^{\infty}p(B)p\left(\max (b_1, \cdots , b_{(N-1)})=y \delta +v_{min} \right) \nonumber\\  
 p(\# \mbox{\small PA}=N),\; L \geq x > y \geq 0, \\ 
q_{xx}=  \sum_{N=1}^{\infty}p(B)p\left(\max (b_1, \cdots , b_{(N-1)})\geq x \delta +v_{min} \right) \nonumber\\
  p(\# \mbox{\small PA}=N),\; L \geq x \geq 0. 
	\end{gather}	
		
Note that the expressions for $z_{xy},q_{xy}$ can be obtained similarly as in~\eqref{eq:cite1},\eqref{eq:olgooo}. Deriving the entering probabilities is the same as obtaining the winner PAs' payments distribution for a second-price auction. Mathematically,\footnote{The symbol $\exists!$ reads as: there is one and only one. Note that it is assumed that at least two bidders participate in the market.}
  \begin{gather}
\Scale[0.78]{\Pi_k^{0}=\sum_{N=2}^{\infty} p(\# \mbox{\small PA}=N)\Bigg[ p\bigg(\{\exists !\; w\in\{1,\cdots,N\} | b_w>k\delta+v_{min}\}} \nonumber\\  ,  
\Scale[0.8]{\{\max(b_I)=k\delta+v_{min}, I\in\{1,\cdots,N\}/w\}\bigg)}\nonumber\\ 
\Scale[0.8]{+\sum_{m=2}^{N} p\bigg(\big\{\exists! \;E\subset \{1,\cdots,N\}, |E|= m \;| \forall e\in E \; b_e=k\delta+v_{min}}\nonumber \\  
\Scale[0.8]{, \forall I\in \{1,\cdots,N\}/E \; b_I<k\delta+v_{min}\big\}\bigg) \Bigg],\; L \geq k\geq 0},  
  \end{gather}
where the first term corresponds to the case in which there is one and only one bid higher than $v_{min}+k\delta$, while the maximum of other bids is equal to $v_{min}+k\delta$, and the second term corresponds to the case in which there are at least two bids equal to $v_{min}+k\delta$ while the others are less than $v_{min}+k\delta$. We obtain the closed-form expressions in this case as follows:
  \begin{itemize}
  \item \textbf{Uniform bids}: ($C(N,n)={N\choose n}$)
  \begin{gather}   \Pi_k^{0}=\sum_{N=2}^{\infty}\Bigg[ \frac{e^{-\lambda} \lambda^N}{(N-1)!} \left(1-\left(\frac{k+1}{L+1}\right)\right) \nonumber \\ 
   \left( \frac{(k+1)^{N-1} - k^{N-1}}{(L+1)^{N-1}}\right) + \frac{e^{-\lambda} \lambda^N}{N!} \nonumber \\  
  \sum_{n=2}^{N} C(N,n) \left(\frac{1}{L+1} \right)^n \left(\frac{k}{L+1} \right)^{N-n}\Bigg].
\end{gather}
\item \textbf{Concatenated sampled Laplace bids}: 
 \begin{gather}   \Pi_k^{0}= \sum_{N=2}^{\infty}\Bigg[  \left(1-P_B(b=v_{min}+k \delta)\right) \nonumber \\ 
  \Big( (p_B(b \leq v_{min}+k \delta))^{N-1} - \nonumber \\  
  (p_B(b \leq v_{min}+k \delta-\delta))^{N-1} \Big) \frac{e^{-\lambda} \lambda^N}{(N-1)! }\nonumber \\  
  +\sum_{n=2}^{N} C(N,n) \left(p_B(b=v_{min}+k \delta)\right)^n \nonumber \\ 
  (p_B(b < v_{min}+k \delta))^{N-n} \left(\frac{e^{-\lambda} \lambda^N}{N!}\right)\Bigg]. 
\end{gather}
\hspace{-1mm}Using~\eqref{Eq:FL}, this expression can be further expanded.
  \end{itemize}
\subsubsection{Expected income of the $\mbox{\small CCN}$ manager in the second-price OBSAs}
Similar to the derivations for the first-price OBSAs in subsection~\ref{subsec:fp}, the $\mbox{\small CCN}$ manager's expected income for the second-price OBSAs can be derived as as follows.
\begin{proposition} The $\mbox{\small CCN}$ manager's expected income in the second-price OBSA when the observer $\mbox{\small PA}$'s residual patience time is $\Delta$ can be obtained as follows:
\begin{gather}\label{eq:rev}
 E[\textrm{Revenue}]=\sum_{i=0}^{L}\sum_{j=0}^{i} j\Theta^i_j \Pi^0_i=\sum_{i=0}^{L} \sum_{j=0}^{i} j \left(Q^{\Delta}\right)_{ij}  \Pi^0_{i}  \nonumber\\ 
+ \sum_{k=1}^{\Delta} \sum_{i=0}^{L} \sum_{j=0}^{i} j \left(Q^{\Delta-k} Z P^{k-1}\right)_{ij}  \Pi^0_{i}. 
\end{gather}
Note that using the defined $Q$, $Z$, $P$, the exact value of the CCN manager's expected income can be obtained. However, expanding the $\Theta^i_j$ given in the following while neglecting the higher order terms leads to an approximated lower bound of the CCN manager's income.
\vspace{-4mm}
\begin{gather}
	\Theta^i_j=P_i(\textrm{Ending at $j^{th}$ state})= 	
		 \Bigg[ \bigg[ \sum_{k=1}^{\Delta} q_{ij}q_{jj}^{\Delta - k}q_{ii}^{k-1}\nonumber\\ 
 + \sum_{k=0}^{\Delta-2} \sum_{c=0}^{\Delta-2-k} \sum_{r=j+1}^{i-1} q_{ir}q_{rj}q_{jj}^{\Delta-k-2-c}q_{i,i}^k q_{r,r}^c  \nonumber
 + \cdots \bigg] \nonumber
 \end{gather}
 \begin{gather}
 + \bigg[ \sum_{k=1}^{\Delta} z_{ij}p_{jj}^{\Delta - k}q_{ii}^{k-1}\nonumber\\ 
+ \sum_{k=0}^{\Delta-2} \sum_{c=0}^{\Delta-2-k} \sum_{r=j+1}^{i-1} q_{ir}z_{rj}p_{jj}^{\Delta-k-2-c}q_{i,i}^k q_{r,r}^c  
  + \cdots \bigg] \Bigg] . 
	\end{gather}
	\end{proposition}

\subsection{The Allowable Patience Time of the $\mbox{\small PAs}$ in OBSAs}

This subsection is devoted to finding the relationship between the $\mbox{\small PAs}$' patience time and the $\mbox{\small CCN}$ manager's profit. As can be seen from Figure~\ref{fig:5}, all the states of the proposed DTHMC are transient except the last state in the first-price OBSAs. This fact implies that the corresponding steady state distribution takes the value of one only at the last state and zero everywhere else. Hence, if the PAs' patience time is sufficiently large, they usually (in the mean sense) pay the $\mbox{\small CCN}$ manager an amount close to $v_{min}$, which is not profitable for the CCN manager. Considering this fact, imposing some constraints on the PAs' patience time is necessary. 

Let $\left(X\right)_{n\geq 0} $ be a DTHMC with state space $S$. The \textit{hitting time} for an arbitrary state $a \in S$ is defined as a random variable $H^{a}=\inf \{n\geq 0 : X_n = a\}$. Subsequently, the \textit{mean hitting time} of an arbitrary state $a$ given an starting state $i$ can be expressed as $k_i^{a}= E(H^a | X_0=i)$.
By inspecting the transition diagram in Figure~\ref{fig:5}, the mean hitting times for state $0$ can be derived as \cite{ref:book}: $k_0^0=0,\; k_1^0=1+ p_{10}k_0^0+  p_{11}k_1^0,\;k_2^0=1+p_{20}k_0^0+p_{21}k_1^0+p_{22}k_2^0 ,\;\cdots  $, which can be written in a compact matrix form as follows:
 \begin{equation}\label{mat:1}
 \Scale[0.7]{\begin{pmatrix} 
 	1 & 0 &0&0& \cdots & 0 \\ 
 	-p_{10} & 1-p_{11} & 0 & 0&\cdots  & 0 \\
	-p_{20}& -p_{21} & 1-p_{22} & 0 &  \cdots & 0 \\
	\vdots & \vdots & \vdots & \vdots & \vdots & \vdots \\
	-p_{L,0} & -p_{L,1} & -p_{L,2} & \cdots & -p_{L,L-1} & 1-p_{LL}
 \end{pmatrix} }
 \Scale[0.7]{\begin{pmatrix} 
 k_0^0 \\
 k_1^0\\
 k_2^0\\
 \vdots \\
 k_L^0 \end{pmatrix} =  
 \begin{pmatrix} 
 0 \\
 1 \\ 
 1 \\
 \vdots \\
 1 \end{pmatrix}}.
 \end{equation}
Similarly, the following equation holds for the second-price OBSAs depicted in Figure~\ref{fig:6}.
 
 \begin{gather}
 \Scale[0.7]{\begin{pmatrix} 
 	1 & 0 &0&0& \cdots & 0 \\ 
 	-q_{10} & 1-q_{11} & 0 & 0&\cdots  & 0 \\
	-q_{20}& -q_{21} & 1-q_{22} & 0 &  \cdots & 0 \\
	\vdots & \vdots & \vdots & \vdots & \vdots & \vdots \\
	-q_{L,0} & -q_{L,1} & -q_{L,2} & \cdots & -q_{L,L-1} & 1-q_{LL}
 \end{pmatrix} 
 \begin{pmatrix} 
 \rho_0^0 \\
 \rho_1^0\\
 \rho_2^0\\
 \vdots \\
 \rho_L^0 \end{pmatrix} +} \\ \nonumber
 \Scale[0.7]{\begin{pmatrix} 
 	-z_{00} & 0 &0&0& \cdots & 0 \\ 
 	-z_{10} & -z_{11} & 0 & 0&\cdots  & 0 \\
	-z_{20}& -z_{21} & -z_{22} & 0 &  \cdots & 0 \\
	\vdots & \vdots & \vdots & \vdots & \vdots & \vdots \\
	-z_{L,0} & -z_{L,1} & -z_{L,2} & \cdots & -z_{L,L-1} & -z_{LL}
 \end{pmatrix} 
 \begin{pmatrix} 
 k_0^0 \\
 k_1^0\\
 k_2^0\\
 \vdots \\
 k_L^0 \end{pmatrix}
 =  
 \begin{pmatrix} 
 0 \\
 1 \\ 
 1 \\
 \vdots \\
 1 \end{pmatrix}},
 \end{gather}
 where $\rho^0_{i}$ is the mean hitting time in the primary Markov chain.
 Using the mean hitting times for state 0, for the second-price OBSAs, the expected absorption time $P^0$ to this state is given by:
 \begin{equation}\label{eq:javab}
P^0= \sum_{i=0}^{L} \Pi _i^0 \rho_i^0.
 \end{equation}
In a similar fashion, for the first-price OBSAs, solving \eqref{mat:1} along with using the vector $\pi^0$ provides the expected absorption time to state $0$. This parameter ($P^0$) can be used for setting an upper bound on the $\mbox{\small PAs}$' patience time. For instance, if the average patience time $\bar{\Delta}$ of the $\mbox{\small PAs}$ is greater than or equal to $P^0$, the $\mbox{\small CCN}$ manager makes a low long-term profit. Hence, a CCN manager can set tighter bounds on the PAs' average patience time or the PAs' maximum patience time $\Delta_{max}$, e.g., $\Delta_{max}\leq \frac{P^0}{m} \; , m\in \mathbb{Z}>1$, to guarantee a higher profit.

 In the above derivations, the focus is on the absorption to state $0$ which implies no profit for the CCN manager. However, with the same logic and the same procedure, one can derive similar equations for hitting the other states (such as state $1$ or $2$) to set tighter upper bounds on the CCN manager's profit.
\section{Interactions between CCN Managers and CPs}\label{sec:interacCCNandCPs}
In this section, we analyze the interactions between the CCN managers and the CPs. We consider a market in which CCN managers are the buyers and CPs are the sellers. The CCN managers compete to acquire resources from the CPs to fulfill their demands and attract more customers. This is a dynamic market in the sense that both of the CCN managers and the CPs can leave and join the market over time. 

For each type of servers, we divide the market into two sub-markets: flat-price market and auction-based market. In the flat-price market, resources are sold with a constant price. In the auction-based market, resources are sold through auctions. Consequently, each CCN manager has two options to obtain extra resources: (i) to buy the resources with a unit flat-price; (ii) to participate in the auctions. Correspondingly, the CPs have two options to sell their extra resources. In the rest of the discussion, without loss of generality, we focus on one type of servers.

It is assumed that auctions are held with a Poisson rate $\lambda_A$, and the number of CCN managers available in the market during each round of auction is Poisson distributed with mean $\lambda_{CCN}$. In each auction, a portion of the CCN managers act as active bidders and participate in the auction, while the rest act as passive participants and do not submit their bids; the fraction of active CCN managers is denoted by $\mu$. Moreover, we assume that the number of CPs participating in each round of the auction follows Poisson distribution with mean $\lambda_{CP}$. The time window in which a CCN manager can stay in the auction market (participation time) is limited to $T_p$. When the participation time of a CCN manager ends, he has to obtain the resources through the flat-price market with price $\varpi$.

In each auction, interested CCN managers participate with their bids reflecting their valuations for the servers, and CPs participate with their offered prices determining their least expected prices to lend their resources. The winner of each auction is the CCN manager who offers the highest bid, and he receives the resources from the CP with the lowest offered price. This happens if the lowest offered price is less than the winner's bid. Otherwise, there is no winner in the auction. As can be seen, two competitions emerge in this market: (i) the competition between CCN managers who try to offer higher prices to win the resources; (ii) the competition between CPs who try to offer lower prices to sell their resources. For the rest of this study, we focus on the CCN managers' competition since the competition among the CPs can be similarly analyzed.

Consider one of the CCN managers with $r$ units of residual participation time. Similar to~\cite{ref:RS}, aggressive bidding strategy is considered, where the CCN managers bid more values as they approach the deadline. In this case, the instantaneous bid of a CCN manager can be modeled as follows:
\begin{equation}\label{eq:mainbidding}
b(r)=e^{-\gamma r}u- D(r),
\end{equation}
where $u$ stands for the instant utility gained from the resource, $\gamma$ the CCN manager's \textit{rate of time preference}, and $D(r)$ the \textit{discounted expected utility}. More precisely, $D(r)$ determines the cost of neglecting the potential future discounts by waiting in the auction. Note that $b(r)$ should be a decreasing function with respect to $r$ since the CCN manager bids higher as he approaches the deadline. 

When a CCN manager's participation time ends, he has to leave the auction and pay the constant price $\varpi$ for the resource. This fact implies the following boundary conditions:

\begin{equation}
\begin{aligned}\label{eq:bc}
b(0)&=\varpi, \\
D(0)&= u-\varpi.
\end{aligned}
\end{equation}

It is assumed that the CCN managers' bids and the CP's offered prices are kept private. Hence, the CCN managers have no knowledge about the bids of each other. Nonetheless, the following information is assumed known to the CCN managers: (i) the cdf $F(b)$ and thus the pdf $f(b)$ of the bids; (ii) the cdf $G(o)$ and thus the pdf $g(o)$ of the CPs' offered prices. 

As mentioned before, there are two well-known mechanisms for determining the winner's payment in the auction, first-price and second-price mechanisms. In-depth analysis of the both mechanisms in the described markets is presented below.\footnote{A simplified version of our analysis is presented in \cite{ref:RS} for a retail market, where the focus is on the second-price payment mechanism. In their model, the offered prices of the sellers are not incorporated into the analysis which simplifies their model and derivations.} 
\subsection{First-price Analysis}
In the first-price payment scheme, the winner CCN manager has to pay the value of his bid to the seller CP. In the following theorem, we obtain the bidding strategy of a CCN manager in a stable market setting where the distributions of the bids and offered prices are stationary.
\begin{theorem}\label{th:be1}
For the described market, with the first-price payment mechanism, a CCN manager's expected utility can be described using the following Hamilton-Jacobi-Bellman (HJB) equation:
\begin{equation}
D(r) = (u-\varpi)e^{-\gamma r}.
\end{equation}
In this case, the bid of a CCN manager with the residual time $r$ in the stable market setting is given by:
\begin{equation}
b^*(r)=\varpi e^{-\gamma r}.
\end{equation}
\end{theorem}
\begin{proof}
Consider a short period of time $\delta$ where the probability of an auction occurring during the period can be expressed as $\delta \mu \lambda_A < 1$. The expected utility of a CCN manager with residual waiting time $r$ is given by:
\begin{equation}
\begin{aligned}
D(r)=& \frac{1}{1+\gamma \delta} \Bigg[\\
& \Bigg(1-\mu \lambda_A \delta\underbrace{\sum_{n=0}^{\infty} \frac{e^{-\lambda_{CCN}}\lambda_{CCN}^n}{n !}  F^n(b(r))}_{(a)}\\
&  
\underbrace{\sum_{n=1}^{\infty} \frac{e^{-\lambda_{CP}}\lambda_{CP}^n}{n !}  \left(1-G^n(b(r))\right)}_{(b)} \Bigg)  D(r-\delta) 
 \\
 +&\mu \lambda_A \delta \Bigg( \bigg[ \sum_{n=0}^{\infty} \frac{e^{-\lambda_{CCN}}\lambda_{CCN}^n}{n !} F^n(b(r))  \\
&\sum_{n=1}^{\infty} \frac{e^{-\lambda_{CP}}\lambda_{CP}^n}{n !}  \left(1-G^n(b(r)) \right) \bigg]\left(e^{-\gamma r}u\right)   \\ 
&- \bigg[ \sum_{n=0}^{\infty} \frac{e^{-\lambda_{CCN}}\lambda_{CCN}^n}{n !} F^n(b(r))  \\
&\sum_{n=1}^{\infty} \frac{e^{-\lambda_{CP}}\lambda_{CP}^n}{n !}  \left(1-G^n(b(r)) \right) \bigg ] b(r) \Bigg) \Bigg], 
\end{aligned}\label{eq:Bel}
\end{equation}
where the factor $1+\gamma \delta$ is the interest of time of the CCN manager. In the above equation, the second and the third lines correspond to the situation in which the CCN manager does not obtain the resource in the current period, and thus he bids in the next time period. The term (a) indicates the probability of winning of a CCN manager when he faces $n$ opponents in the auction. The term (b) indicates encountering at least a CP with a lower offered price than the current bid of the CCN manager. The rest of the equation describes the situation in which the CCN manager wins the auction in the current period. In this case, there would be a change in the utility, described by the fourth and the fifth lines, and an incurred payment which is described in the last two lines. The terms (a) and (b) can be simplified as follows: 
\begin{align}
\sum_{n=0}^{\infty} \frac{e^{-\lambda_{CCN}}\lambda_{CCN}^n}{n !} F^n\left(b(r)\right) &= e^{\lambda_{CCN} ( F(b(r))-1) },
\end{align}
\begin{align}
\sum_{n=1}^{\infty} \frac{e^{-\lambda_{CP}}\lambda_{CP}^n}{n !}  \left(1-G^n(b(r)\right)& = 1- e^{\lambda_{CP} \left(G(b(r))-1\right)}.  
\end{align}
Using the above expressions and applying some algebraic manipulations,~\eqref{eq:Bel} can be rewritten as:
\begin{equation}
\begin{aligned}\label{eq:lastExp}
\gamma &D(r)= \frac{D(r-\delta) - D(r)}{\delta} + \mu \lambda_A  \\
& \Bigg( e^{\lambda_{CCN} \left( F(b(r))-1 \right)} \left(1- e^{\lambda_{CP} \left(G(b(r))-1\right)}\right)  \\
&\left(-D(r-\delta) +e^{-\gamma r} u-b(r)\right)\Bigg). 
\end{aligned}
\end{equation}
The result of the theorem can be obtained by letting $\delta \rightarrow 0 $ and applying the boundary conditions given in~\eqref{eq:bc}.

\end{proof}
\subsection{Second-price Analysis}
In the second-price mechanism, the winner CCN manager has to pay the second highest bid to the seller CP. To obtain the CCN managers' bidding strategy in the stable market setting, we first derive the HJB equation for this case in the following theorem. Then, we propose a method to solve the derived HJB equation in the following corollaries.
\begin{theorem}\label{th:bell3}
For the described market with the second-price payment mechanism, in the stable market setting, a CCN manager's expected utility can be described using the following HJB equation:
\begin{equation}\label{Eq:theoSecond}
\begin{aligned}
\gamma D(r)&= -D'(r) + \mu \lambda_A  \\
 &\Bigg( e^{\lambda_{CCN} ( F(b(r))-1)} \left(1- e^{\lambda_{CP} (G(b(r))-1)}\right)  \\
&\left(e^{-\gamma r} u -D(r) \right)-\left(1- e^{\lambda_{CP} (G(b(r))-1)}\right)  \\
& \bigg[e^{-\lambda_{CCN}}b(T) + \int_{b(T)}^{b(r)} \lambda_{CCN} \\ 
 & e^{\lambda_{CCN} (F(b(t))-1)}b(t)F'(b(t))d (b(t))  \bigg] \Bigg). 
\end{aligned}
\end{equation}

\end{theorem}
\begin{proof}
We can follow the procedure described in the proof of Theorem~\ref{th:be1} while changing the expected payment term in the derivations. In the second-price payment mechanism, a CCN manager's expected payment can be expressed as:
\begin{gather}
\left(1- e^{\lambda_{CP} (G(b(r))-1)}\right) \Big[e^{-\lambda_{CCN}}b(T) + \nonumber \\ 
 \sum_{n=1}^{\infty} \frac{e^{-\lambda_{CCN}}\lambda_{CCN}^n}{n !} \int_{b(T)}^{b(r)} b(t) n F^{n-1}(b(t)) F'(b(t))dt  \Big],
\end{gather}
where the first term expresses the probability of encountering at least one CP with a lower offered price compared to the CCN manager's bid, and the terms inside the bracket indicate the expected payment of the CCN manager upon encountering no opponent (the first term) or at least one opponent (the second term). Note that $n F^{n-1}(b(t)) F'(b(t))$ is the pdf of the second highest bid in an auction. In this case, the expected utility can be derived as follows:
 \begin{equation}
 \begin{aligned}\label{eq:lastExp2}
\gamma D(r)&= \frac{D(r-\delta) - D(r)}{\delta} + \mu \lambda_A  \\
 &\Bigg[ e^{\lambda_{CCN} ( F(b(r))-1)} \left(1- e^{\lambda_{CP} (G(b(r))-1)}\right)  \\
&\left(-D(r-\delta) +e^{-\gamma r} u\right)  \\
&-\left(1- e^{\lambda_{CP} (G(b(r))-1)}\right) \\ 
& \Big[e^{-\lambda_{CCN}}b(T) +  \int_{b(T)}^{b(r)} \lambda_{CCN}   \\
&e^{\lambda_{CCN}(F(b(t))-1)}b(t)F'(b(t))d(b(t))  \Big] \Bigg]. 
\end{aligned}
 \end{equation}
The result of the theorem can be obtained by letting~$\Scale[0.95]{\delta \rightarrow 0}$.
\end{proof}

\begin{corollary}\label{th:DE}
The bid of a CCN manager with the residual participation time $r$, in the second-price payment mechanism, can be described by the following integro-differential equation:
\begin{equation}
 \begin{aligned}\label{eq:lem1}
 b^*(r)& \bigg[-\gamma -\mu \lambda_A\left(e^{\lambda_{CCN} \left( F(b^*(r))-1\right)}\right) \\ 
& \left(1- e^{\lambda_{CP} (G(b^*(r))-1)} \right) \bigg] -
 b^{*^{'}}(r)  \\ 
 +& \mu \lambda_A \lambda_{CCN} \left(1- e^{\lambda_{CP} (G(b^*(r))-1)}\right)  \\  
 &\int_{b^*(T)}^{b^*(r)}  e^{\lambda_{CCN} (F(b(t))-1)}b(t)F'\left(b(t)\right)d(b(t))  \\
 +& \mu \lambda_A  \left(1- e^{\lambda_{CP} (G(b^*(r))-1)}\right) e^{-\lambda_{CCN}}b^*(T)   \\ 
 -&\gamma u e^{-\gamma r}  =0 .
 \end{aligned}
\end{equation}
 Considering the boundary conditions given in~\eqref{eq:bc}, one can numerically solve the above equation so as to find the optimal bidding strategy at each time instant.
\end{corollary}
\begin{proof}
Replacing $D(r)$ and $D'(r)$ in~\eqref{Eq:theoSecond} with~\eqref{eq:mainbidding} yields:
\begin{equation}
\begin{aligned}
\gamma \big(&e^{-\gamma r}u-b(r)\big)= \left(b'(r)+\gamma e^{-\gamma r}u\right) + \mu \lambda_A   \\
& \Bigg[ e^{\lambda_{CCN} ( F(b(r))-1)} \left(1- e^{\lambda_{CP} (G(b(r))-1)}\right) b(r)  \\
&-\left(1- e^{\lambda_{CP} (G(b(r))-1)}\right) \Big[e^{-\lambda_{CCN}}b(T) +  \\ 
& \int_{b(T)}^{b(r)} \lambda_{CCN} e^{\lambda_{CCN} (F(b(t))-1)}b(t)F'(b(t))d (b(t))  \Big] \Bigg].
 \end{aligned}
\end{equation}

 By rearranging the terms, one can obtain~\eqref{eq:lem1}.
\end{proof}
Note that no assumption is made on the pdf and cdf of the bids and offered prices in the above derivations and the corresponding results hold for any arbitrary distribution. In the following, by assuming a specific cdf and pdf for these quantities, we explicitly derive the bidding strategy of a CCN manager as a first-order differential equation.
\begin{corollary} \label{col:bidOpt}
Consider a market in which higher bids and offered prices happen with low probabilities. More precisely, assume that the pdfs of bids and offered prices measured at each value $x$ are both inversely proportional to $x$. In this case, the pdf and cdf of the bids and the offered prices can be described as follows:
\begin{equation}
\begin{aligned}
g(b(r))&=f(b(r))= \frac{1}{b(r) [\ln(z)-\ln(a)]} ,  \\ 
G(b(r))&=F(b(r))= \frac{\ln (b(r)) - \ln (a)}{\ln(z)-\ln(a)},  \\ 
b(r)&\in [a,z],\; 0<a<z<\varpi  .
\end{aligned}\label{eq:bid11}
\end{equation}
If the initial CCN manager's bid is small such that $b(T)\rightarrow a$, where $a \rightarrow 0$, then the bidding strategy of a CCN manager, with residual patience time $r$, is dictated by the following first-order differential equation:
\begin{equation}
\begin{aligned}\label{eq:firstdiff}
y'(r)+c_1(y(r))^{c_2}+c_3(y(r))^{c_4}+c_5(y(r))^{c_6}+c_7 y(r) \\ 
=-\gamma u e^{-\gamma r},
\end{aligned}
\end{equation}
where $y(r)=b^*(r)$, and 
\begin{equation}
 \begin{aligned}\label{eq:constantsofdiff}
c_1=& -\mu \lambda_A e^{-\lambda_{CCN}} e^{\frac{-\lambda_{CCN} \ln(a)}{\ln(z)-\ln(a)}} e^{-\lambda_{CP}} e^{\frac{-\lambda_{CP} \ln(q)} {\ln(z)-\ln(q)}} \\ 
&+  \mu \lambda_A \lambda_{CCN} e^{-\lambda_{CP}}e^{\frac{-\lambda_{CP} \ln(q)}{\ln(z)-\ln(q)}}  \\
 & e^{-\lambda_{CCN}} e^{\frac{-\lambda_{CCN} \ln(a)}{\ln(z)-\ln(a)}}\left(\frac{1}{\lambda_{CCN}+\ln(z)-\ln(a)}\right), \\
c_2=& \frac{\lambda_{CCN}}{\ln(z)-\ln(a)} + \frac{\lambda_{CP}}{\ln(z)-\ln(q)}+1,  \\
c_3=& -\gamma e^{-\lambda_{CP}} e^{\frac{-\lambda_{CP} \ln(q)}{\ln(z)-\ln(q)}},  \\
c_4=&\frac{\lambda_{CP}}{\ln(z)-\ln(q)}+1, \\
c_5=& \mu \lambda_A  e^{-\lambda_{CCN}} e^{\frac{-\lambda_{CCN} \ln(a)} {\ln(z)-\ln(a)}},   \\
&+ \mu \lambda_A \lambda_{CCN} e^{-\lambda_{CCN}}e^{\frac{-\lambda_{CCN} \ln(a)}{\ln(z)-\ln(a)}}  \\
c_6=& \frac{\lambda_{CCN}}{\ln(z)-\ln(a)}+1 , \\
c_7=&\gamma  .
 \end{aligned}
\end{equation}
\end{corollary}
\begin{proof}
Considering the given expressions for the cdf and pdf of bids and offered prices,~\eqref{eq:lem1} can be rewritten as:
\vspace{-2mm}
\begin{equation}
 \begin{aligned}
 b(r) &\Bigg[-\gamma -\mu \lambda_A e^{-\lambda_{CCN}}e^{\frac{\lambda_{CCN} \ln (b(r))}{\ln(z)-\ln(a)}} e^{\frac{-\lambda_{CCN} \ln(a)}{\ln(z)-\ln(a)}}  \\ 
 &\left(1-e^{-\lambda_{CP}}e^{\frac{\lambda_{CP} \ln (b(r))}{\ln(z)-\ln(q)}} e^{\frac{-\lambda_{CP} \ln(q)}{\ln(z)-\ln(q)}}\right) \Bigg] -
b'(r)  \\ 
& + \mu \lambda_A \lambda_{CCN} \left(1-e^{-\lambda_{CP}}e^{\frac{\lambda_{CP} \ln (b(r))}{\ln(z)-\ln(q)}} e^{\frac{-\lambda_{CP} \ln(q)}{\ln(z)-\ln(q)}}\right)  \\  
&\Bigg( \int_{b(T)}^{b(r)}  e^{-\lambda_{CCN}}e^{\frac{\lambda_{CCN} \ln (b(t))}{\ln(z)-\ln(a)}} e^{\frac{-\lambda_{CCN} \ln(a)}{\ln(z)-\ln(a)}}   \\ 
&\frac{1}{ \ln(z)-\ln(a)} d(b(t)) \Bigg)  \\
 &+ \mu \lambda_A  \left(1-e^{-\lambda_{CP}}e^{\frac{\lambda_{CP} \ln (b(r))}{\ln(z)-\ln(q)}} e^{\frac{-\lambda_{CP} \ln(q)}{\ln(z)-\ln(q)}}\right)  \\ 
&  e^{-\lambda_{CCN}}b(T) -\gamma u e^{-\gamma r}  =0.
 \end{aligned}
\end{equation}
Performing some algebraic manipulations and taking the integral and then rearranging the terms yields:
\vspace{-1.5mm}
\begin{equation}
 \begin{aligned}
-&\gamma (b(r)) + \gamma e^{-\lambda_{CP}} e^{\frac{-\lambda_{CP} \ln(q)}{\ln(z)-\ln(q)}} (b(r))^{\frac{\lambda_{CP}}{\ln(z)-\ln(q)}+1}  \\
-&\mu \lambda_A  e^{-\lambda_{CCN}} e^{\frac{-\lambda_{CCN} \ln(a)} {\ln(z)-\ln(a)}} (b(r))^{\frac{\lambda_{CCN}}{\ln(z)-\ln(a)}+1}   \\
+& \mu \lambda_A e^{-\lambda_{CCN}} e^{\frac{-\lambda_{CCN} \ln(a)}{\ln(z)-\ln(a)}} e^{-\lambda_{CP}} e^{\frac{-\lambda_{CP} \ln(q)} {\ln(z)-\ln(q)}}
 \\
 &(b(r))^{\frac{\lambda_{CCN}}{\ln(z)-\ln(a)} + \frac{\lambda_{CP}}{\ln(z)-\ln(q)}+1}
 - b'(r)  \\ 
 +& \mu \lambda_A \lambda_{CCN} e^{-\lambda_{CCN}}e^{\frac{-\lambda_{CCN} \ln(a)}{\ln(z)-\ln(a)}}  \\
  &(b(r))^{\frac{\lambda_{CCN}}{\ln(z)-\ln(a)}+1} \left(\frac{1}{\lambda_{CCN}+\ln(z)-\ln(a)}\right) \\
  -& \mu \lambda_A \lambda_{CCN} e^{-\lambda_{CP}}e^{\frac{-\lambda_{CP} \ln(q)}{\ln(z)-\ln(q)}}  \\
 & (b(r))^{\frac{\lambda_{CP}}{\ln(z)-\ln(a)}} e^{-\lambda_{CCN}} e^{\frac{-\lambda_{CCN} \ln(a)}{\ln(z)-\ln(a)}}  \\
  &\left(\frac{1}{\lambda_{CCN}+\ln(z)-\ln(a)}\right) (b(r))^{\frac{\lambda{CCN}}{\ln(z)-\ln(a)}+1} \\
 +& \mu \lambda_A - \mu \lambda_A e^{-\lambda_{CP}}e^{\frac{-\lambda_{CP} \ln(q)}{\ln(z)-\ln(q)}}   (b(r))^{\frac{\lambda_{CP}}{\ln(z)-\ln(q)}}  \\
 &e^{-\lambda_{CCN}}b(T) -\gamma u e^{-\gamma r}  =0.  
 \end{aligned}
\end{equation}
This can be expressed by~\eqref{eq:firstdiff}-\eqref{eq:constantsofdiff}.
\end{proof}

\vspace{-5mm}
\section{Simulation Results}\label{sec:sim}
\noindent The performance of the proposed scheme is illustrated through four simulation scenarios.
\subsection{Scenario 1: CCN Managers' Income in OBSAs}
We consider the market described in Table~\ref{table:0}, where the bids' range is adopted from the history of winners' bid of Amazon's memory optimized instance (x1.32xlarge) in the past three month~\cite{ref:amazonPricing}. The minimum (maximum) value of bids corresponds to the $50\%$ of the minimum winner bid (maximum winner bid) in that dataset upon requesting the cloud instance for $24$ hours.\footnote{In the second-price case, the bumped $\mbox{\small PA}$ is assumed to participate in the next auction with the probability $p(B)=0.5$.} The presented analytic results for the $\mbox{\small CCN}$ manager's income in Section~\ref{sec:analysisOBSA} and the results obtained from $10000$ Monte Carlo simulations are depicted in Figure~\ref{sim:4}. It can be seen that the analytic results and the simulation results agree well with each other. Also it can be observed that the $\mbox{\small CCN}$ manager's income becomes less sensitive to the number of $\mbox{\small PAs}$ after the market gets sufficiently crowded. Moreover, the effect of increasing the patience time on the $\mbox{\small CCN}$ manager's income becomes more pronounced when there are fewer participant $\mbox{\small PAs}$. 
\begin{table}[H]
	\centering
		\resizebox{\columnwidth}{!}{%
	\begin{tabular}{|c |c |c |c| }
		\hline
	 Auction mechanism&	Bids' Range & Distribution & Parameters \\ \hline
	\begin{tabular}{c}Fist-price  \\ \hline Second-price \end{tabular} &	[48,312] & \begin{tabular}{c}Uniform  \\ \hline Sampled Laplace \end{tabular} &\begin{tabular}{c} -  \\ \hline $\mu=70,\delta=1,w=50$ \end{tabular}  \\ \hline
	\end{tabular}}
	\caption{Simulation setting for the dynamic market between the PAs and the CCN managers.}\label{table:0}
\end{table}
\subsection{Scenario 2: Bids of CCN Managers}
To demonstrate the CCN managers' bidding behavior in our proposed model, we simulate the result of Corollary~\ref{col:bidOpt}. For this purpose, a market with the parameters described in Table~\ref{table:1} is considered. As before, considering the pricing history of Amazon's memory optimized instance in the past three month~\cite{ref:amazonPricing}, the value of $z$ is obtained assuming the maximum bid of a CCN for the cloud resources to be one-third of its maximum sold price. Due to the lack of a real dataset, the rest of the parameters are chosen based on common sense. Figure~\ref{sim:bids} depicts the value of the CCN managers' bids over time for various values of $\gamma$. As can be seen, as the value of $\gamma$ increases, the CCN managers lose their interest in buying resources faster and bidding higher values. 

\begin{table}[H]
\centering
 \begin{tabular}{|c |c |c| c|c|c|c|c|}
 \hline
 Parameter & $u$ & $\mu$ &$\lambda_A$&$\lambda_{CCN}$&$\lambda_{CP}$ &$z$&$a$\\ \hline
 Value & $5$ & $0.6$ & $0.2$ & $0.5$ & $0.75$ & $104$   & $0.01$	 \\ \hline
 \end{tabular}
 \caption{Simulation setting for the dynamic market between the CCN managers and the CPs.}\label{table:1}
\end{table}

\subsection{Scenario 3: Selling the Cloud Resources}
In Figure~\ref{sim:7}, the improvement that can be obtained using the OBSAs instead of the sequential combinatorial auctions is shown in terms of the total number of winners and total sold cloud servers. Also, Figure~\ref{sim:10} reveals the corresponding CCN manager's income. For this simulation, $10$ types of servers are considered where the $\mbox{\small CCN}$ manager holds $1000$ auctions for each of these types synchronously at every minute, where the total number of needed types of servers for each PA is uniformly distributed in the interval $[1,10]$. The arrival rates (per minuet) for servers are i.i.d and uniformly distributed in the interval $[1/15,1]$. This interval is derived based on the reported data in~\cite{ref:googleTrace}, which indicates that the majority of task executions are under $15$ minutes. It is assumed that customers do not leave the market without obtaining their demanded resources (long patience time). The results in Figures~\ref{sim:7},~\ref{sim:10} indicate an increase in the customers' satisfaction, utilization of resources, and the CCN manager's income upon using the OBSAs. 
\subsection{Scenario 4: Market Stability}
Figure~\ref{sim:8} compares the variance of the winners' payment for the second-price OBSA described in Scenario 1 and that of the sequential combinatorial auction for various number of PAs ($N$) and different residual patience times ($\Delta$). Less variance for the payments in the former is seen in Figure~\ref{sim:8} implying a more predictable market for the $\mbox{\small CCN}$ manager. Also, similar results are observed for the first-price OBSAs, which are omitted for the interest of space. 

Figure~\ref{sim:9} depicts the number of participant PAs in the second-price OBSA and that in the sequential combinatorial auction. In this scenario, at each time instant one auction occurs, where the number of available PAs at each time instant is assumed to be a Poisson random variable with mean $93$. This value is adopted from~\cite{ref:googleTrace} assuming that $10\%$ of all the PAs participate in each round of auction. Each PA can delay his participation. It is assumed that $20\%$ of these PAs have side information about the future market situation. It is assumed that the market receives the lowest bids from participant PAs in every $10$ time instances. As can be seen, by utilizing sequential combinatorial auctions, those informed PAs delay their entrance into the market, leading to an unfavorable burst of arrival. In contrast, this issue is effectively suppressed by the proposed OBSAs.  
\section{Conclusion and Future Work}~\label{sec:concl}
\hspace{-2mm}In this work, we proposed a comprehensive two stage framework to describe resource allocation and gathering in modern cloud networks. The first stage describes the interactions between the PAs and the CCN managers. For this stage, OBSAs along with their theoretical analysis are proposed, which enjoy a simple winner determination process and provide the truthfulness property. The second stage models the interactions between the CCN managers and the CPs. For this stage, a theoretical framework is developed to model the bidding behavior of the CCN managers. For future work, one worthwhile direction is to explore the optimization of the social welfare or other parameters of interest. Studying the resource allocation and the load balancing problems jointly is also interesting. In this case, a CCN manager should consider the geographical locations of the servers and $\mbox{\small CPs}$ to find the optimal resource allocation. 

  \begin{figure*}[t]
	\minipage{5.5cm}
	\includegraphics[width=1.0\linewidth, height=0.7\linewidth]{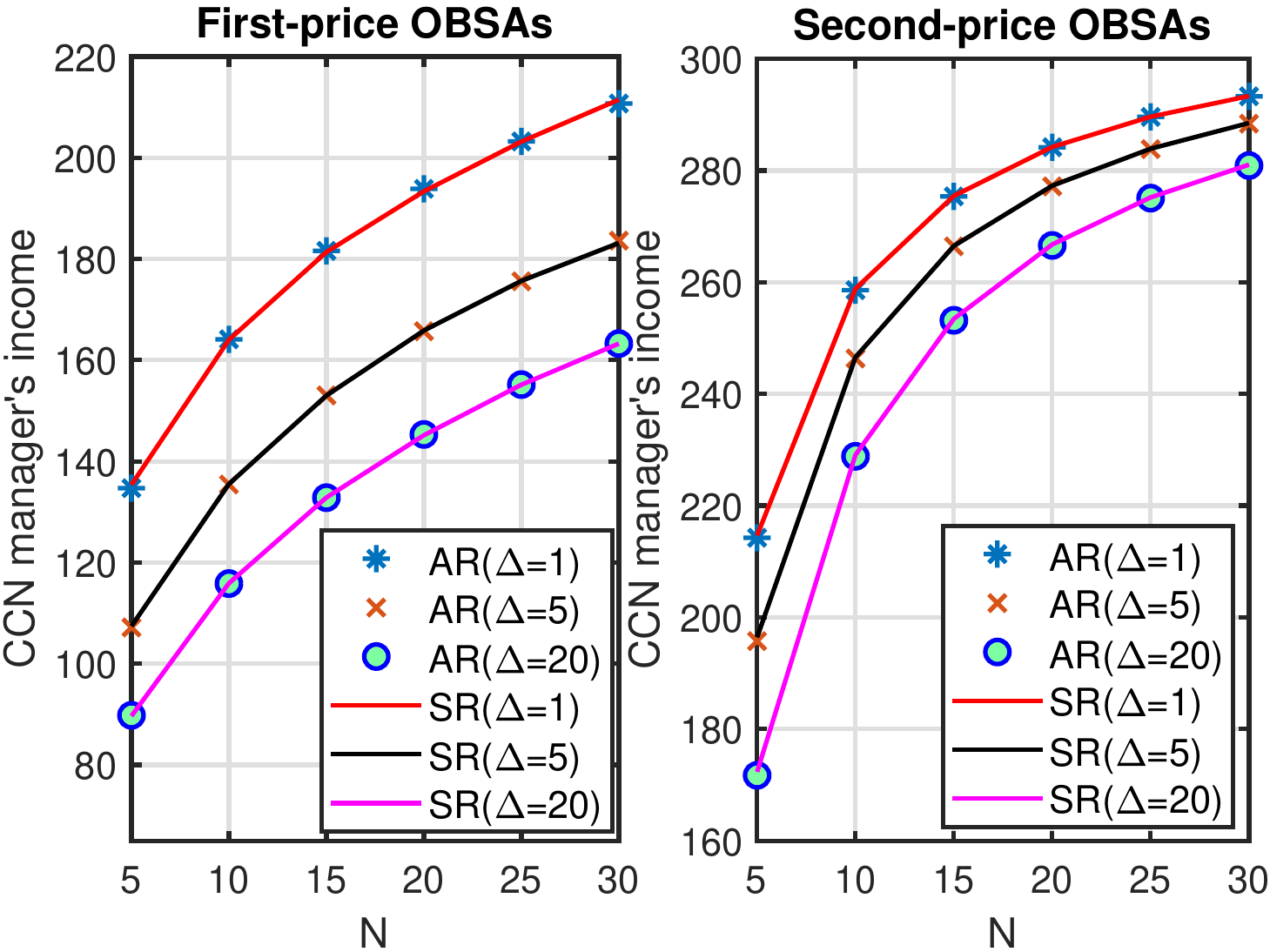}
	\caption{Analytic results ($\mbox{\small AR}$) and simulation results ($\mbox{\small SR}$) for the $\mbox{\small CCN}$  manager's income in the first- and second-price OBSAs for various number of $\mbox{\small PAs}$ ($\mbox{\small N}$) and different residual patience times ($\Delta$).\label{sim:4}}
	\endminipage
	\quad
	\minipage{5.5cm}
		\vspace{-7.9mm}
	\includegraphics[width=1.0\linewidth, height=0.7\linewidth]{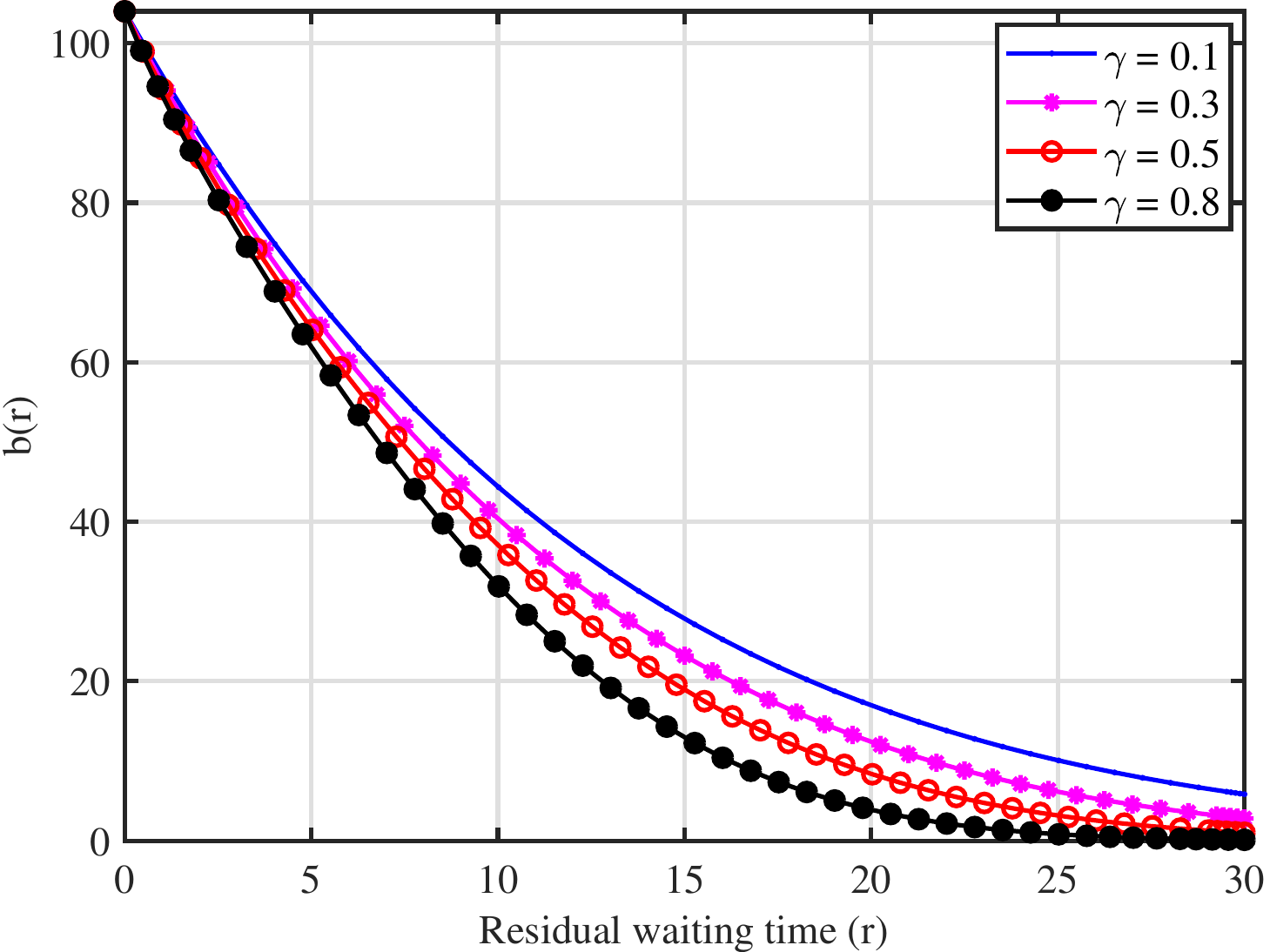}
\caption{The CCN managers' bids with respect to the residual waiting time for different rates of time
preference.\label{sim:bids}}
	\endminipage
	\quad
	\minipage{5.5cm}
    	\vspace{-1.0mm}
	\includegraphics[width=1.0\linewidth, height=0.7\linewidth]{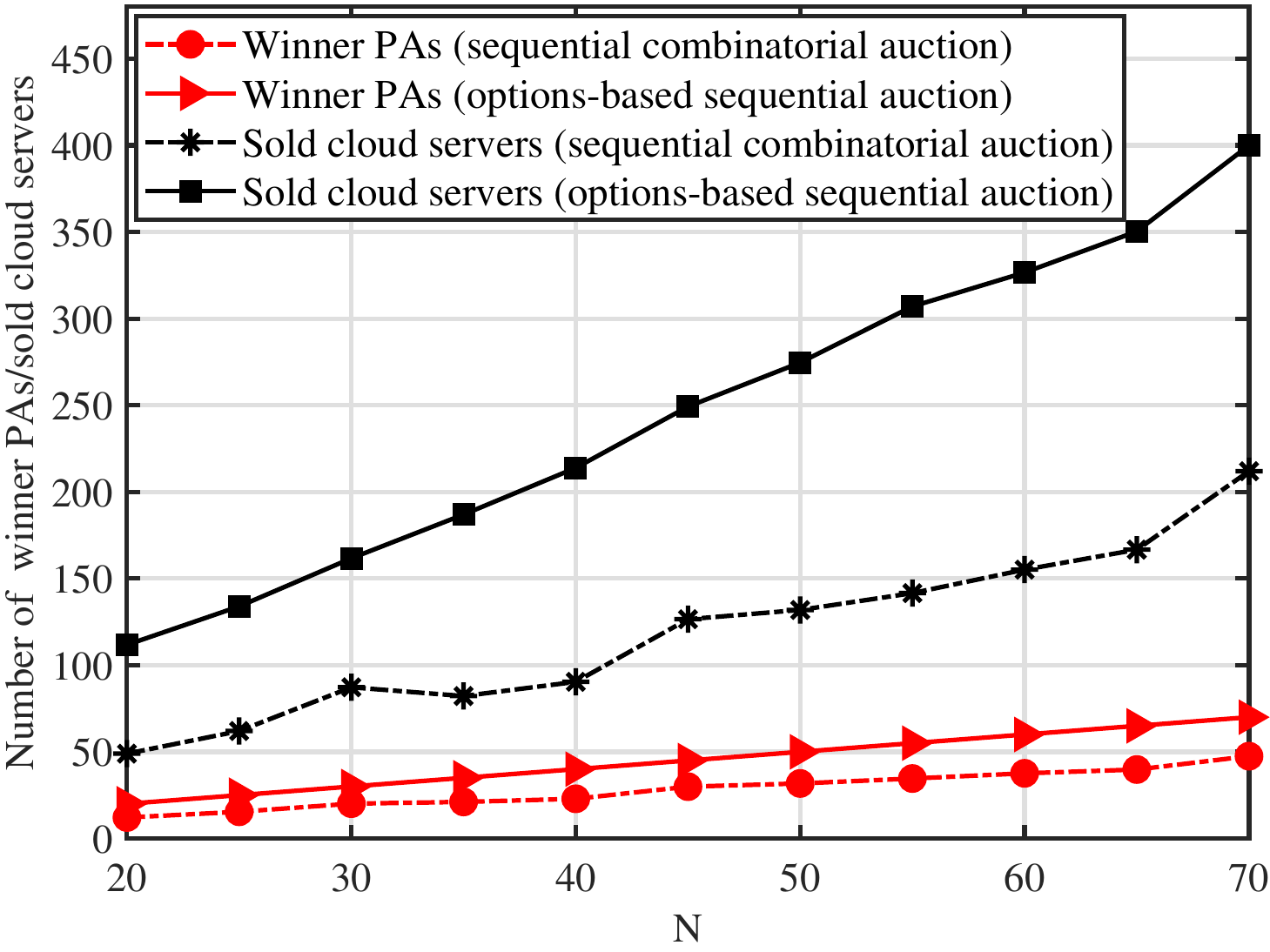}
	\caption{Performance comparison between the sequential combinatorial auction and the second-price OBSA for various number of $\mbox{\small PAs}$ in the system ($\mbox{\small N}$).\label{sim:7}}
	\endminipage
\end{figure*}

  \begin{figure*}[t]
	\minipage{5.5cm}
    	\vspace{-4mm}
	\includegraphics[width=1.0\linewidth, height=0.7\linewidth]{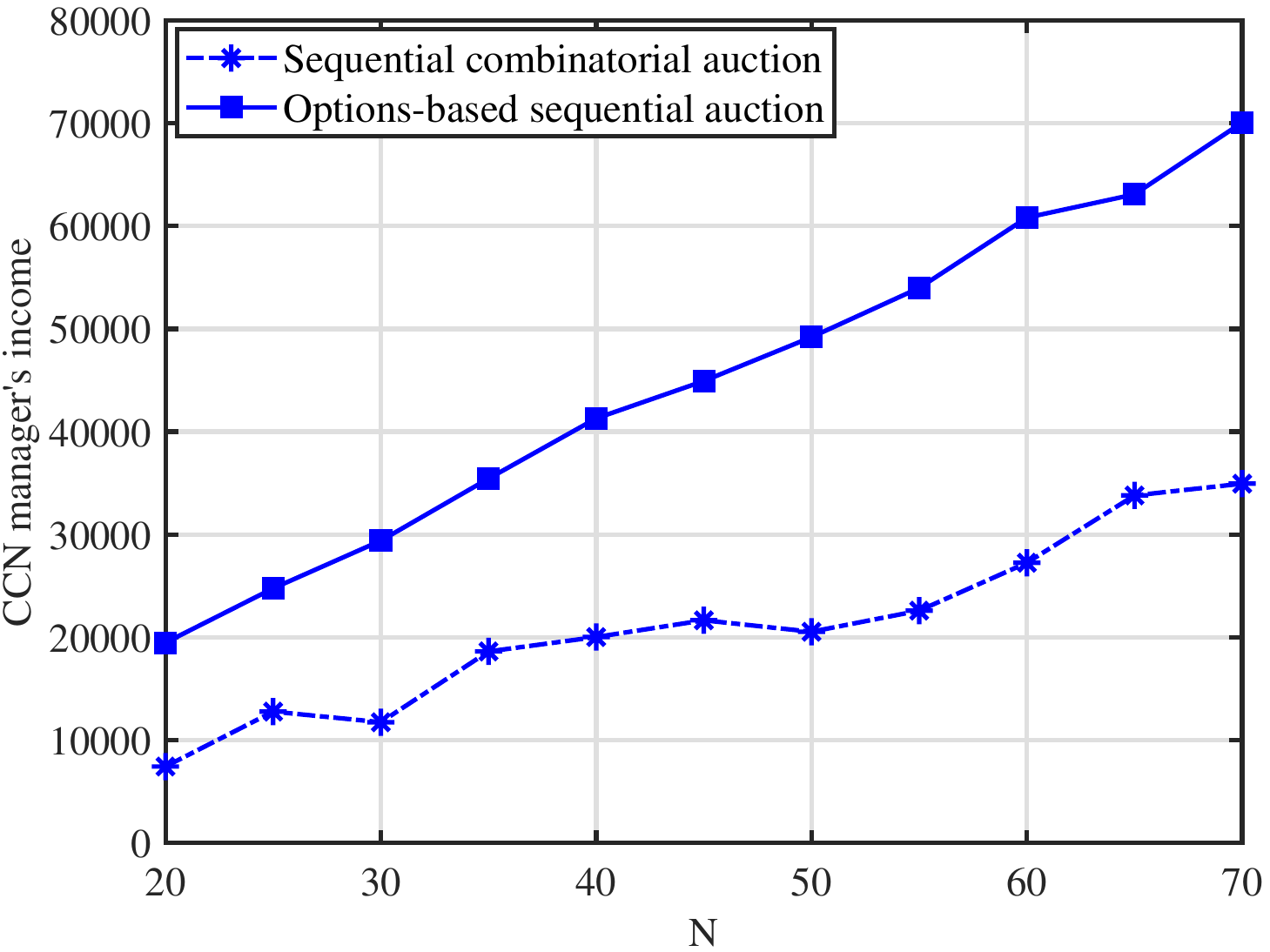}
	\caption{Comparison between the CCN manager's income upon utilizing sequential combinatorial auction and the second-price OBSA for various number of $\mbox{\small PAs}$ in the system ($\mbox{\small N}$).\label{sim:10}}
	\endminipage
	\quad
	\minipage{5.5cm}
		\vspace{-1mm}
	\includegraphics[width=1.0\linewidth, height=0.7\linewidth]{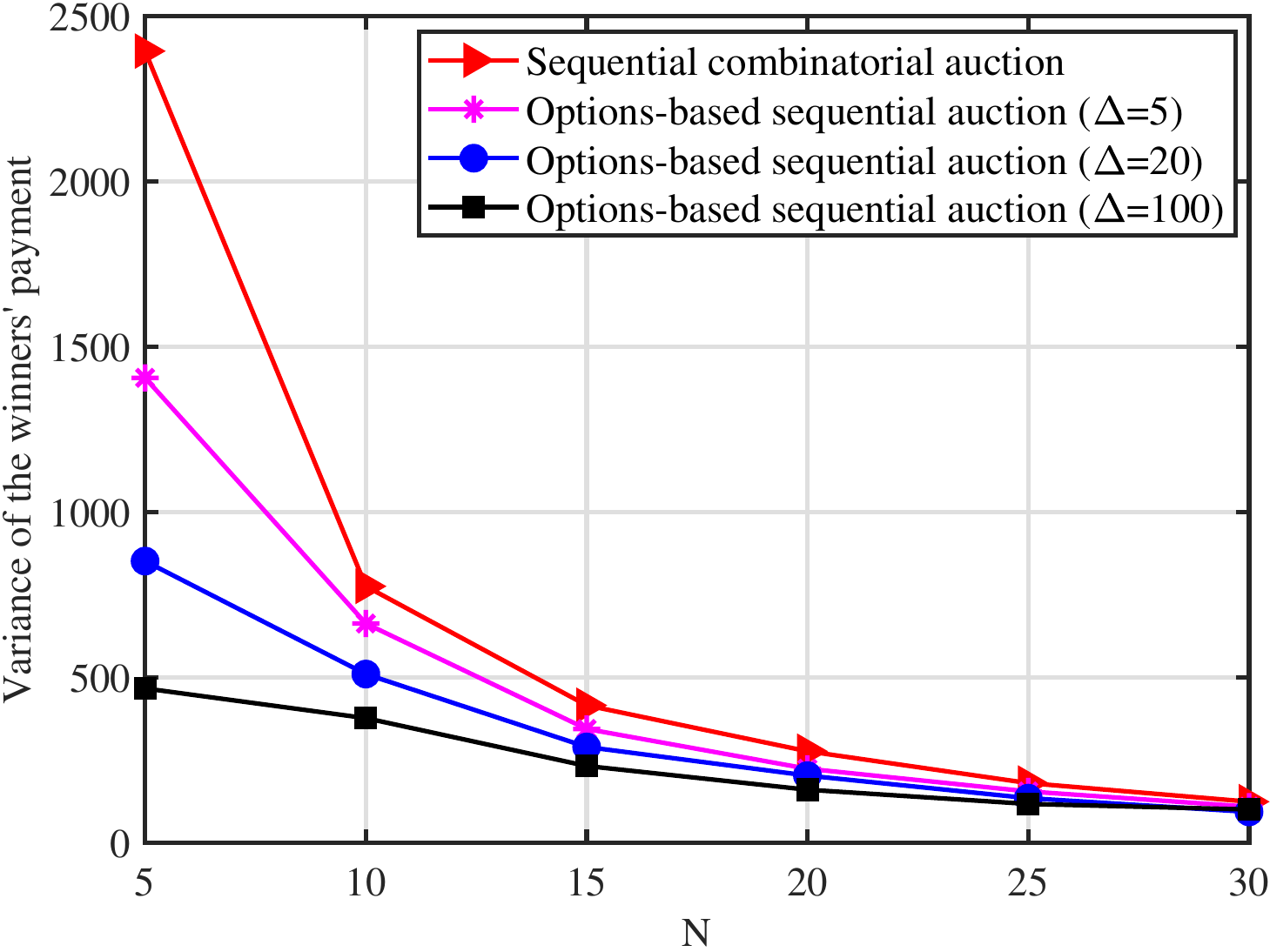}
\caption{Comparison between the variance of the winners' payment in the second-price OBSA vs. sequential combinatorial auction for various number of $\mbox{\small PAs}$ ($\mbox{\small N}$) and different residual patience times ($\Delta$).\label{sim:8}}
	\endminipage
	\quad
	\minipage{5.5cm}
			\vspace{-8.15mm}
	\includegraphics[width=1.0\linewidth, height=0.7\linewidth]{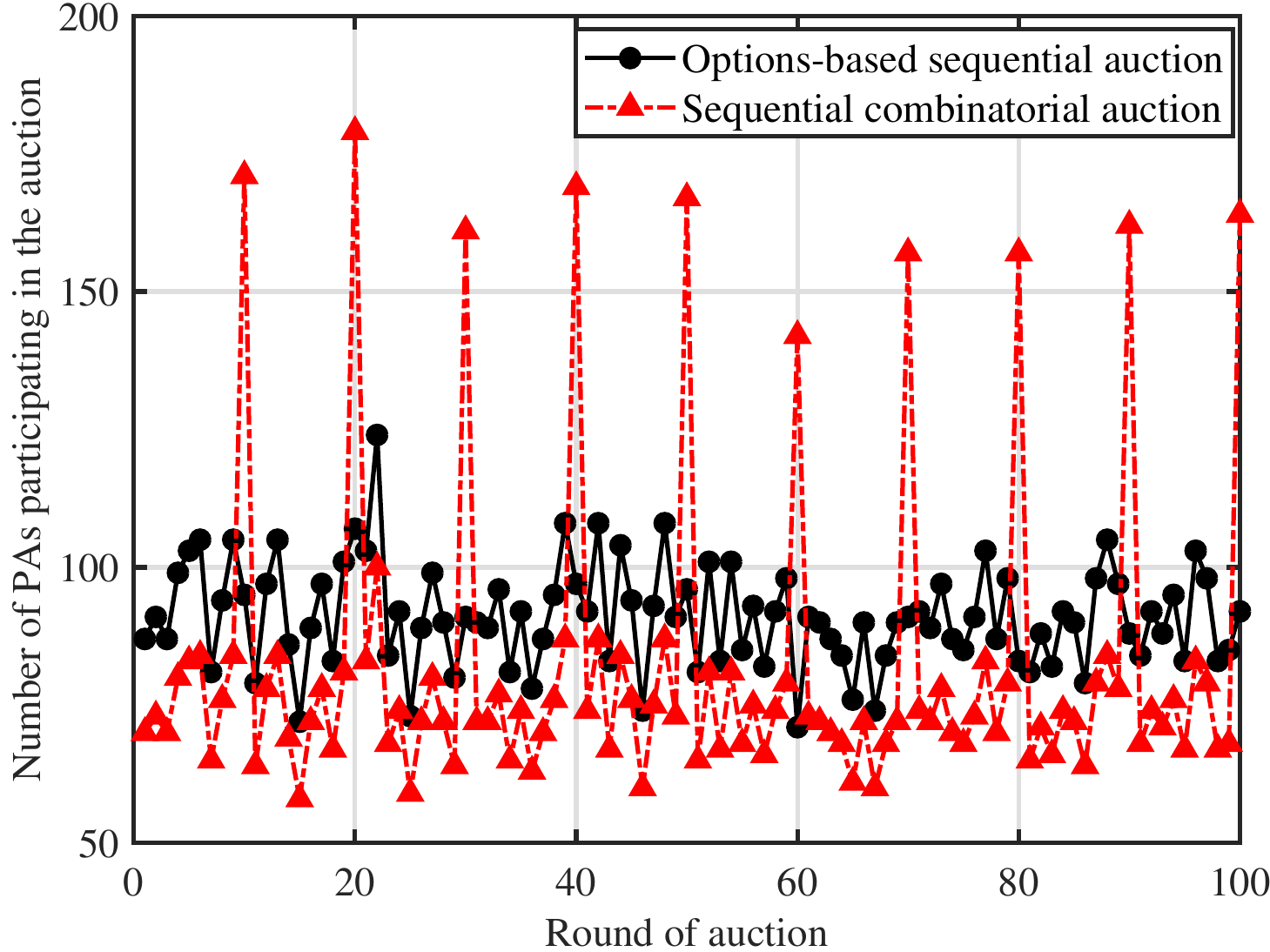}
\caption{Comparison between the number of participant PAs in the second-price OBSA vs. sequential combinatorial auction.\label{sim:9}}
	\endminipage
\end{figure*}
\bibliographystyle{IEEEtran}
\bibliography{BibTCC}
\begin{IEEEbiographynophoto}{Seyyedali Hosseinalipour} (S'17)
 received the B.S. degree in Electrical Engineering from Amirkabir University
of Technology (Tehran Polytechnic), Tehran, Iran in 2015. He is currently pursuing a Ph.D. degree in
the Department of Electrical and Computer Engineering at North Carolina State University, Raleigh, NC, USA. His research interests include analysis of wireless networks, resource allocation and load balancing for cloud networks and resource allocation and task scheduling for vehicular ad-hoc networks.
\end{IEEEbiographynophoto}
\begin{IEEEbiographynophoto}{Huaiyu Dai} (F’17)
  received the B.E. and M.S. degrees in electrical engineering from Tsinghua
University, Beijing, China, in 1996 and 1998, respectively, and the Ph.D. degree in electrical
engineering from Princeton University, Princeton, NJ in 2002.

He was with Bell Labs, Lucent Technologies, Holmdel, NJ, in summer 2000, and with AT\&T Labs-Research, Middletown, NJ, in summer 2001. He is currently a Professor of Electrical and Computer
Engineering with NC State University, Raleigh. His research interests are in the general areas of
communication systems and networks, advanced signal processing for digital communications, and
communication theory and information theory. His current research focuses on networked information
processing and crosslayer design in wireless networks, cognitive radio networks, network security, and
associated information-theoretic and computation-theoretic analysis.

He has served as an editor of IEEE Transactions on Communications, IEEE Transactions on Signal
Processing, and IEEE Transactions on Wireless Communications. Currently he is an Area Editor in
charge of wireless communications for IEEE Transactions on Communications. He co-edited two
special issues of EURASIP journals on distributed signal processing techniques for wireless sensor
networks, and on multiuser information theory and related applications, respectively. He co-chaired the
Signal Processing for Communications Symposium of IEEE Globecom 2013, the Communications
Theory Symposium of IEEE ICC 2014, and the Wireless Communications Symposium of IEEE
Globecom 2014. He was a co-recipient of best paper awards at 2010 IEEE International Conference on
Mobile Ad-hoc and Sensor Systems (MASS 2010), 2016 IEEE INFOCOM BIGSECURITY
Workshop, and 2017 IEEE International Conference on Communications (ICC 2017).
\end{IEEEbiographynophoto}
\end{document}